%% file: paper.tex
\documentclass[11pt]{article}
\usepackage{amsmath,amssymb,amsfonts,amsthm,epsfig}
\usepackage[usenames,dvipsnames]{xcolor}
\usepackage{bm,xspace}
\usepackage{tcolorbox}
\usepackage{cancel}
\usepackage{fullpage}
\usepackage{liyang}
\usepackage{framed}
\usepackage{verbatim}
\usepackage{enumitem}
\usepackage{array}
\usepackage{multirow}
\usepackage{afterpage}
\usepackage{mathrsfs}
\usepackage{pifont} 
\usepackage{chngpage}
\usepackage[normalem]{ulem}
\usepackage{boxedminipage}
\usepackage{caption}
\usepackage{subcaption}
\usepackage{bold-extra}
\usepackage{forest}

\usepackage{algorithm}
\usepackage{algorithmicx}
\usepackage{algpseudocode}

\usepackage{tikz}
\usetikzlibrary{calc,through,backgrounds,decorations.pathreplacing, calligraphy,arrows.meta}
\usetikzlibrary{positioning,chains,fit,shapes}

\usepackage{tikz-cd}

\usepackage{pgfplots}
\pgfplotsset{width=8cm,compat=newest}

\newcommand{\lnote}[1]{\footnote{{\bf \color{blue}Li-Yang}: {#1}}}

\def\colorful{0}

\ifnum\colorful=1
\newcommand{\violet}[1]{{\color{violet}{#1}}}
\newcommand{\orange}[1]{{\color{orange}{#1}}}
\newcommand{\blue}[1]{{{\color{blue}#1}}}
\newcommand{\red}[1]{{\color{red} {#1}}}

\fi
\ifnum\colorful=0
\newcommand{\violet}[1]{{{#1}}}
\newcommand{\orange}[1]{{{#1}}}
\newcommand{\blue}[1]{{{#1}}}
\newcommand{\red}[1]{{{#1}}}

\fi

\newcommand{\isedge}{{\sc IsEdge}} 
\newcommand{\vcover}{{\sc VertexCover}} 
\def\SAT{{\sc SAT}}

\newcommand{\NP}{\mathrm{NP}} 
\newcommand{\Ptime}{\mathrm{P}} 
\newcommand{\VC}{\mathrm{VC}}

\newcommand{\Rel}{\mathrm{Rel}}

\newcommand{\dtsize}{\mathrm{DT}} 
\newcommand{\ind}{\mathrm{Ind}}

\newlist{enumprop}{enumerate}{1} % set up a dedicated enumeration environment
\setlist[enumprop]{label=\arabic*.,ref=\theproposition.\arabic*}
\crefalias{enumpropi}{proposition}

\makeatletter
\newtheorem*{rep@theorem}{\rep@title}
\newcommand{\newreptheorem}[2]{
\newenvironment{rep#1}[1]{
 \def\rep@title{#2 \ref{##1}}
 \begin{rep@theorem}\itshape}
 {\end{rep@theorem}}}
\makeatother

\newreptheorem{theorem}{Theorem}

\begin{document}

\input{intro}

\input{preliminaries}
\input{main_reduction}
\input{hardness_distillation}
\input{constant_error}

 \section*{Acknowledgments}

 We thank Pasin Manurangsi for a helpful conversation and the FOCS reviewers for their comments and feedback.  

The authors are supported by NSF awards 1942123, 2211237, 2224246 and a Google Research Scholar award. Caleb is also supported by an NDSEG fellowship, and Carmen by a Stanford Computer Science Distinguished Fellowship.  

\bibliography{ref}
\bibliographystyle{alpha}

% \appendix
% \input{appendix}

\end{document}

%% file: intro.tex
%!TEX root = paper.tex

\title{Properly learning decision trees with queries is NP-hard \vspace{10pt}}

\author{ 
Caleb Koch \vspace{6pt} \\ 
\hspace{-10pt} { {\sl Stanford}} \and 
\hspace{5pt} Carmen Strassle \vspace{6pt} \\
 { {\sl Stanford}} \vspace{15pt}
 \and 
Li-Yang Tan \vspace{6pt}  \\
\hspace{-10pt} {{\sl Stanford}}
}

\date{\small{\today}}

 \maketitle

\begin{abstract}
We prove that it is NP-hard to properly PAC learn decision trees with queries, resolving a longstanding open problem in learning theory (Bshouty 1993; Guijarro--Lavín--Raghavan 1999; Mehta--Raghavan 2002; Feldman 2016). While there has been a long line of work, dating back to (Pitt–Valiant 1988), establishing the hardness of properly learning decision trees from {\sl random examples}, the more challenging setting of {\sl query} learners necessitates different techniques and there were no previous lower bounds. En route to our main result, we simplify and strengthen the best known lower bounds for a different problem of Decision Tree Minimization (Zantema–-Bodlaender 2000; Sieling 2003).

On a technical level, we introduce the notion of {\sl hardness distillation}, which we study for decision tree complexity but can be considered for any complexity measure: for a function that requires large decision trees, we give a general method for identifying a small set of inputs that is responsible for its complexity. Our technique even rules out query learners that are allowed constant error. This contrasts with existing lower bounds for the setting of random examples which only hold for inverse-polynomial error.

\violet{ Our result, taken together with a recent almost-polynomial time query algorithm for properly learning decision trees under the uniform distribution (Blanc--Lange--Qiao--Tan 2022), demonstrates the dramatic impact of distributional assumptions on the problem. } 
 \end{abstract} 

\thispagestyle{empty}
\newpage

\tableofcontents
\thispagestyle{empty}
\newpage 

\setcounter{page}{1}

\section{Introduction} 

Decision trees are among the most basic and popular hypothesis classes in machine learning. They have long served as the gold standard of interpretability: a classic, influential survey of statistical models states that ``On interpretability, decision trees rate an A+"~\cite{Bre01twocultures}, and two decades later, a survey of intepretable machine learning~\cite{RCCHSZ22} lists the optimization of decision tree hypotheses as the very first of the field’s ``10 grand challenges". Besides interpretability, decision tree hypotheses are extremely fast to evaluate, with evaluation time scaling with their depth, a quantity that is often exponentially smaller than their overall size. Decision trees are also at the heart of powerful ensemble methods such as random forests and XGBoost which achieve state-of-the-art performance across a variety of domains. 

We consider the task of constructing {\sl optimal} decision tree representations of data. A standard formalization of this task is the problem of {\sl properly PAC learning decision trees}: given access to a target function $f$ and a distribution $\mathcal{D}$, construct the optimal decision tree hypothesis for~$f$ under~$\mathcal{D}$. Valiant’s original definition of the PAC model~\cite{Val84} considered learners with both passive access to the target function in the form of random labeled examples as well as active access in the form of queries. This setting as well as that of learning from random examples only have since been intensively studied. Valiant’s motivation for the more powerful query setting came from modeling interactions with an expert (``[an] important aspect of the formulation is that the notion of oracles makes it possible to discuss a whole range of teacher-learner interactions beyond the mere identification of examples"). The query setting also models the task of converting an existing accurate but inscrutable hypothesis, for which one has query access to, into a more intelligible representation—once again, decision trees are a canonical sought-for representation for this task~\cite{CS95,BS96,AB07,ZH16,BKB17,VLJODV17,FH17,VS20}.  

\paragraph{This work.} The NP-hardness of properly learning decision trees from {\sl random examples} is a foundational result known since the early days of PAC learning~\cite{Ang,PV88}. The question of whether there exists an efficient {\sl query} learner, on the other hand, has been raised repeatedly over the years, in research papers~\cite{Bsh93, GLR99, MR02} and surveys~\cite{Fel16}. We resolve this question by showing that properly learning decision trees is NP-hard even for query learners:  \medskip 

\begin{tcolorbox}[colback = white,arc=1mm, boxrule=0.25mm]
\begin{theorem}\label{thm:main-intro}  There is an absolute constant $\eps > 0$ such that the following holds. Suppose there is an algorithm that, given queries to an n-variable function f computable by a decision tree of size $s = O(n)$ and random examples $(\bx,f(\bx))$ drawn according to a distribution $\mathcal{D}$, runs in time $t(n)$ and w.h.p.~outputs a size-s decision tree h that is $\eps$-close to f under $\mathcal{D}$. Then $\mathrm{SAT} \in \mathrm{RTIME}(\poly(t(n^2\polylog n)))$.
\end{theorem} %\calnote{There needs to be a poly factor inside $t$. The exact result is $\mathsf{SAT}\in O(n^2\polylog (n)\cdot t(n^2\polylog n))$. We could state it as ``\SAT\ can be solved in randomized $\Tilde{O}(n^2\cdot t(n^2))$-time''. Or $\Tilde{O}(\poly(t(n^2))$-time. Or $\Tilde{O}(t(n^2)^2)$-time.}
\end{tcolorbox}
\medskip

\Cref{thm:main-intro} addresses a stark gap in our understanding of the problem. The fastest known algorithm runs in exponential time, $2^{O(n)}$ for all values of $s$. There were no previous lower bounds, leaving open the possibility of a $\poly(n,s)$-time algorithm. Indeed, existing query learners for various relaxations of the problem had suggested that such an algorithm was within striking distance. \Cref{thm:main-intro} provides for the first time strong evidence that there are no polynomial-time, or indeed even subexponential-time, algorithms for the problem. 

\subsection{Background and Context} 

\paragraph{Hardness of properly learning decision trees from random examples.} 

NP-hardness in the setting of random examples has been known since the seminal work of Pitt and Valiant~\cite{PV88}. Their paper, which initiated the study of the hardness of proper learning, attributed the result to an unpublished manuscript of Angluin~\cite{Ang}. Subsequently, Hancock, Jiang, Li, and Tromp~\cite{HJLT96} obtained hardness even of {\sl weakly-proper} learning, where the algorithm is allowed to return a decision tree of size larger than that of the target. There have since been several works~\cite{ABFKP09,KST23,Bsh23} further improving~\cite{HJLT96}'s result. 

These works build crucially on a simple reduction from {\sc SetCover}, a reduction variously attributed to Levin~\cite{Lev73}, Angluin~\cite{Ang}, and Haussler~\cite{Hau88}. We describe this technique and discuss why it is limited to the setting of random examples in~\Cref{sec:new techniques necessary}.     

\paragraph{Algorithms for properly learning decision trees.} There is a simple Occam algorithm for properly learning decision trees from random examples: for a size-$s$ decision tree target, draw $O(s \log(n))$ many labeled examples and use dynamic programming to find a size-$s$ decision tree hypothesis that fits the dataset perfectly (see e.g.~\cite{GLR99,MR02}). Standard generalization bounds~\cite{BEHW89} show that this algorithm satisfies the PAC guarantee. Its runtime is $2^{O(n)}$, with the dynamic program being the bottleneck.

Ehrenfeucht and Haussler~\cite{EH89} gave a faster algorithm that runs in time $n^{O(\log s)}$, but their algorithm is only weakly proper: for a size-$s$ target, its hypothesis can be as large as $n^{\Omega(\log s)}$. This large gap is a significant drawback---decision tree hypotheses are interpretable and fast to evaluate insofar as they are {\sl small}---and~\cite{EH89} stated as the first open problem of their paper that of designing algorithms that produce smaller hypotheses. There has been no progress on this problem in the setting of random examples since 1989.   

\paragraph{The power of queries.} In contrast, granting the learner queries enables the design of several polynomial-time algorithms that {\sl almost} solve the problem of properly learning decision trees. Already in his original paper~\cite{Val84} (see also~\cite{Ang88}), Valiant gave a polynomial-time query algorithm for properly learning monotone DNFs; consequently, for size-$s$ monotone decision tree targets Valiant’s algorithm returns a size-$s$ monotone DNF as its hypothesis. Other such results include polynomial-time query learners for general decision tree targets that output depth-$3$ formulas~\cite{Bsh93} and polynomials~\cite{KM93,SS93} as hypotheses. As further demonstration of the power of queries, a recent work of Blanc, Lange, Qiao, and Tan~\cite{BLQT22} gives an almost-polynomial-time ($\poly(n)\cdot s^{O(\log\log s)}$ time) query algorithm that properly learns decision trees under the {\sl uniform distribution}. Finally, the query model opens the possibility of circumventing long-known SQ lower bounds for the problem~\cite{BFJKMR94}, which show that in the setting of random examples all SQ algorithms must take time $n^{\Omega(\log s)}$. 

Taken together, this was all evidence in favor of a polynomial-time, or at least a mildly-super-polynomial time algorithm for properly learning decision trees with queries. In light of \Cref{thm:main-intro}, even a subexponential-time algorithm is now unlikely. 

\subsection{Other related work}

\paragraph{Scarcity of hardness results for PAC learning with queries.} 

\Cref{thm:main-intro} adds to a dearth of NP-hardness results for the model of PAC learning with queries. Indeed, we are aware of only one other such result: in~\cite{Fel06} Feldman proved that DNFs are NP-hard to properly learn with queries, resolving a longstanding problem of Valiant~\cite{Val84,Val85}. As Feldman remarked in his paper, this was the first NP-hardness result, for {\sl any} learning task, for the model of PAC learning with queries. (Our techniques are entirely different from~\cite{Fel06}'s.) 

Related to the scarcity of hardness results, there are numerous query algorithms, for a variety of learning tasks, whose runtimes remain unmatched in the setting of random examples. It is also well known that under standard cryptographic assumptions, PAC learning with queries is strictly more powerful than from random examples only. 

\paragraph{Hardness of properly learning decision trees in other models.}   

While the focus of our work is on the PAC model, interest in the hardness of properly learning decision trees predates and extends beyond the PAC model. An early paper by Hyafil and Rivest~\cite{HR76} proved the NP-hardness of constructing generalized decision trees (ones with more expressive splits than the values of single variables) that perfectly fit a given dataset; quoting the authors, “The importance of this result can be measured in terms of the large amount of effort that has been put into finding efficient algorithms for constructing optimal binary decision trees”. Other results on the hardness of properly learning decision trees in other models include~\cite{GJ79,KPB99,GLR99,ZB00,BB03,LN04,CPRAM07,RRV07,Sie08,AH12,Rav13,BLQT21random}.  

\subsection{Technical remarks about \Cref{thm:main-intro}}

\paragraph{Hardness for constant error.} 

A notable aspect of \Cref{thm:main-intro} is that it rules out learners that are allowed constant error. This was not known even in the setting of random examples, where existing hardness results only hold for inverse-polynomial error: prior to our work, there were no lower bounds ruling out algorithms for properly learning size-s DTs, from random examples only, in time say $(ns)^{O(1/\eps)}$, which is polynomial for constant $\eps$. (Feldman’s NP-hardness result for properly learning DNFs with queries also only holds for inverse-polynomial error.)    

\paragraph{Implications for decision tree minimization.} 

The actual result that we prove is stronger than as stated in \Cref{thm:main-intro}: it holds even if the learner is given explicit descriptions of the target function $f$ and the distribution $\mathcal{D}$ as inputs. Furthermore, the target function can even be given to the learner in the form of a decision tree. For this reason, our result also has implications for the problem of {\sl decision tree minimization}: given a decision tree, find an equivalent one of minimum size. We recover the best known hardness of approximation result for this problem~\cite{ZB00,Sie08} via what is, in our opinion, a much simpler proof. Our proof also yields a stronger result: we show that the problem remains hard even if the resulting tree only has to agree with the original tree on a small given set of inputs. 

\paragraph{Implications for testing decision trees.}

Another aspect in which the actual result we prove is stronger than as stated in~\Cref{thm:main-intro} is that it even rules out distribution-free {\sl testers} for decision trees. (The fact that lower bounds against testers for a class yield lower bounds against proper learners for the same class is well known and easy to show~\cite{GGR98}.) While there’s a large body of work giving lower bounds for testing various classes of functions, the vast majority of these results are information-theoretic in nature, focusing on query complexity, with far fewer computational lower bounds. Our result does not rule out decision tree testers with low query complexity, but it shows that even if such a tester exists, it must nevertheless run in exponential time (unless SAT admits a subexponential time algorithm).

\section{Technical Overview} 

\subsection{Why the query setting necessitates new techniques}
\label{sec:new techniques necessary} 

Before delving into our techniques, we describe the key construction~\cite{Lev73,Ang,Hau88} at the heart of all previous results on the hardness of properly learning decision trees from random examples~\cite{Ang,HJLT96,ABFKP09,KST23,Bsh23} and discuss why it is limited this setting. (This section can be freely skipped; its point is to explain why we had to depart from previous approaches in order to prove \Cref{thm:main-intro}.) 

Consider the following reduction from {\sc SetCover} to the problem of properly learning disjunctions. Let $\mathcal{S} = \{ S_1,\ldots,S_n\}$ be a {\sc SetCover} instance over the universe $[m]$ and define $u^{(1)},\ldots,u^{(m)} \in \zo^n$ where 
\[ (u^{(j)})_i = 
\begin{cases}
    1 & \text{if $j \in S_i$} \\
    0 & \text{otherwise.} 
\end{cases}
\]
Let $C\sse [n]$ be the indices of an optimal set cover for $\mathcal{S}$ and consider the target disjunction $f : \zo^n \to \zo$, 
\[ f = \bigvee_{i \in C} x_i. \]
Let $\mathcal{D}$ be the uniform distribution over $\{ u^{(1)},\ldots,u^{(m)}, 0^n\}.$  Note that given any disjunction hypothesis $h$ for $f$ that achieves error $< 1/(m+1)$ under $\mathcal{D}$, the variables in $h$ must constitute a set cover for $\mathcal{S}$. 

To see why this reduction, and reductions like it, do not extend to the setting of queries, we first observe that this specific target function can be easily learned with queries, simply by querying $f$ on all strings of Hamming weight 1. More generally and crucially, we note that the target function is defined by the optimal solution to the {\sc SetCover} instance. While this is a very natural strategy (and indeed many other hardness results for learning employ such a strategy), for any such reduction it seems challenging to provide query access to the target function without having to solve the {\sc SetCover} instance, which would of course render the reduction inefficient. (Beyond the issue of queries, this reduction is also limited to the inverse-polynomial error regime and does not rule out learners that are allowed larger error.) While this reduction is for the hardness of properly learning disjunctions, all aforementioned hardness results for decision trees use it as their starting point and suffer from the same limitations.

\paragraph{How our approach differs.} Departing from these works, we design a reduction where the {\sl definition} of our target function does not depend on the solution to a computationally hard problem—which allows us to efficiently provide the learner query access to it—and only its {\sl decision tree complexity} scales with the quality of the solution; see~\Cref{remark:queries to target}. Our resulting reduction is quite a bit more elaborate than those for the setting of random examples.
 
\subsection{Overview of our proof and techniques} 

We prove~\Cref{thm:main-intro} by reducing from {\sc VertexCover}: we design an efficient mapping from graphs $G$ to functions $f$ where the decision tree complexity of $f$ reflects the vertex cover complexity of $G$. The properties of this mapping that we require our application to learning are somewhat subtle to state, so we describe and motivate them incrementally. 

\subsubsection{The core reduction} 
\label{sec:intro core} 
Our starting point is a reduction with the following basic properties: \medskip 
\begin{tcolorbox}[colback = white,arc=1mm, boxrule=0.25mm]
\vspace{3pt} 
\begin{center}{\bf The core reduction}
\vspace{-4pt} 
\end{center} 
\begin{itemize}[leftmargin=0.5cm]
\item[$\circ$] Yes case: If $G$ has a small vertex cover, then $f$ has small decision tree complexity.
\item[$\circ$] No case: If $G$ requires a large vertex cover, then $f$ has large decision tree complexity. 
\end{itemize} 
\end{tcolorbox}
\medskip

\renewcommand{\Ind}{\mathrm{Ind}}

We sketch the main ideas behind this core reduction. For an $n$-vertex graph $G$, we consider its  {\sl edge indicator function} $\mathrm{\isedge}_G : \zo^n \to \zo$. An input $v = (v_1,\ldots,v_n) \in \zo^n$ to $\mathrm{\isedge}_G$ is viewed as specifying the presence or absence of each vertex $v_1,\ldots,v_n \in V$, and $\mathrm{\isedge}_G(v) = 1$ iff $v$ specifies the presence of exactly the two endpoints of some edge of $G$. More formally: 

\begin{definition}[$\mathrm{\isedge}_G$]\label{def: isedge}
Let $G = (V,E)$ be an $n$-vertex graph. For an edge $e = \{v_i,v_j\} \in E$, we write $\Ind[e] \in \zo^n$  to denote its indicator string: 
\[ 
\Ind[e]_k = \begin{cases}
1 & \text{if $k\in \{i,j\}$} \\ 
0 & \text{otherwise.}  
\end{cases} 
\]
The {\sl edge indicator function 
of $G$} is the function $\mathrm{\isedge}_G : \zo^n \to \zo$, 
\[ \mathrm{\isedge}_G(v_1,\ldots,v_n) =
\begin{cases} 
1 & (v_1,\ldots,v_n) =\Ind[e] \text{ for some $e \in E$} \\
0 & \text{otherwise.} 
\end{cases}
\]
\end{definition}
When $G$ is clear from context, we drop the subscript and simply write \isedge. 
\paragraph{Warmup.} We first prove: 

\begin{claim}[Decision tree complexity of $\mathrm{\isedge}$]
\label{claim:isedge-intro} 
Let $G$ be an $n$-vertex $m$-edge graph. 
\begin{itemize}
    \item[$\circ$] Yes case: If $G$ has a vertex cover of size $\le k$, then there is a decision tree $T$ for $\mathrm{\isedge}_G$ of size 
    \[ |T| \le k + m + mn. \] 
    \item[$\circ$] No case: If $G$ requires a vertex cover of size $\ge k'$, then any decision tree $T$ for $\mathrm{\isedge}_G$ must have size 
    \[ |T| \ge k' + m. \]
\end{itemize}
\end{claim} 

As stated,~\Cref{claim:isedge-intro} is not useful since the upper bound of the Yes case is much larger than the lower bound of the No case, owing to the additional additive factor of $mn$. More precisely, we need these bounds to satisfy:  
\begin{equation}
\label{eq:Yes less than No} 
\text{If $k' =  (1+\delta)k$ then (Upper bound of Yes case}) < (\text{Lower bound of No case}) \quad  \tag{$\star$} \end{equation}  in order to invoke the NP-hardness of $(1+\delta)$-approximating {\sc VertexCover}.

\paragraph{Amplification.} We therefore consider an ``amplified" version of $\mathrm{\isedge}_G$, \[ \ell\text{-}\mathrm{\isedge}_G : \zo^n \times (\zo^n)^\ell \to \zo, \] and prove: 

\begin{theorem}[Decision tree complexity of $\ell$-$\mathrm{\isedge}$]
\label{thm:ell-isedge-intro} 
Let $G$ be an $n$-vertex $m$-edge graph and $\ell \in \N$. 
\begin{itemize}
     \item[$\circ$] Yes case: If $G$ has a vertex cover of size $\le k$, then there is a decision tree $T$ for $\ell$-$\mathrm{\isedge}_G$ of size 
    \[ |T| \le (\ell+1)\cdot (k + m) + mn. \] 
    \item[$\circ$] No case: If $G$ requires a vertex cover of size $\ge k'$, then any decision tree $T$ for $\ell$-$\mathrm{\isedge}_G$ must have size 
    \[ |T| \ge (\ell+1)\cdot (k' + m). \]
\end{itemize}
\end{theorem}

We point out two properties of~\Cref{thm:ell-isedge-intro} that will be important for us: 

\begin{remark}[Asymmetric amplification in the Yes case]
\label{rem:asymm} 
Comparing~\Cref{claim:isedge-intro} and~\Cref{thm:ell-isedge-intro}, we see that in No case, the entire lower bound of $k' + m$ is amplified by a factor of $\ell+1$. On the other hand, in the Yes case only the $k+m$ factor---and crucially, {\sl not} the $mn$ factor---is amplified by a factor of $\ell+1$. This is important as it  allows us to choose $\ell$ to be sufficiently large to make the $mn$ factor negligible, thereby having our bounds satisfy the sought-for property~(\ref{eq:Yes less than No}). \end{remark} 
\begin{remark}[Efficiently providing query access to $\ell$-$\mathrm{\isedge}_G$]
\label{remark:queries to target} 
We defer the definition of $\ell$-$\mathrm{\isedge}_G$ to~\Cref{sec:ell-isedge} but mention here that (i) it will be the hard target function in our proof of~\Cref{thm:main-intro}; and (ii) just like the unamplified $\mathrm{\isedge}_G$ function---and {\sl unlike} the {\sc SetCover}-based target function described in~\Cref{sec:new techniques necessary}--- its definition will depend only on the edges in $G$ and not its optimal vertex cover. This is crucial as it allows us to efficiently provide the learner query access to its values in our reduction without having to solve {\sc VertexCover}. Circling back to our discussion in~\Cref{sec:new techniques necessary}, this is a key qualitative difference between our reduction and previous reductions for the setting of random examples. 
\end{remark} 

\subsubsection{Hardness distillation} \Cref{thm:ell-isedge-intro} already allows us to recover, with a markedly simpler proof, the best known hardness of approximation result \cite{ZB00,Sie08} for decision tree minimization. However, it does not yet have any  implications for learning since the No case only states that any decision tree that computes $f$ {\sl exactly} must have large size, and does not rule out the possibility that $f$ can be well-approximated by a small decision tree. 

We therefore strengthen the No case via a process that we call {\sl hardness distillation}: we identify a small set of inputs $D\sse \zo^n$, which we call a {\sl coreset}, that is responsible for $f$’s large decision tree complexity. \medskip 

\begin{tcolorbox}[colback = white,arc=1mm, boxrule=0.25mm]
\vspace{3pt} 
\begin{center} 
{\bf The core reduction with hardness distillation}
\vspace{-4pt}
\end{center} 
\begin{itemize}[leftmargin=0.5cm]
\item[$\circ$] Yes case: If $G$ has a small vertex cover, then $f$ has small decision tree complexity.
\item[$\circ$] No case: If $G$ requires a large vertex cover, then there is a small set $D\sse \zo^n$ such that any decision tree that agrees with $f$ on $D$ must be large. 
\end{itemize} 
\end{tcolorbox}
\medskip

Such a reduction yields the NP-hardness of learning decision trees to error $< 1/|D|$, which motivates the problem of constructing coresets that are as small as possible. Our coreset will have size $\poly(n)$, and therefore we get  the hardness of learning to inverse-polynomial error. (In the next subsection we describe a further extension of this technique that establishes constant-error hardness.)  

\paragraph{Hardness distillation via certificate complexity and relevant variables.} We give a general method for identifying a small coreset that witnesses the large decision tree complexity of a function~$f$. At a high level, there are two main components to this coreset: 

\begin{enumerate}
    \item A set of inputs $D_1$ that ensures that any decision tree $T$ that agrees with $f$ on $D_1$ must have a long path $\pi$, one of length at least $s_1$.  
    \item Another set of inputs $D_2$ that ensures that the at-least-$s_1$ many disjoint subtrees that branch off of $\pi$ must have sizes that sum up to at least $s_2$. 
\end{enumerate}
%\lnote{Can we make this figure smaller? It is now taking up half a page, and does not need to be this large to convey the information it is conveying. Don't worry about this if it takes too much time. Another option is to just reference~\Cref{fig:hardness distillation} in the body of the paper.}  
%\input{fig intro}

See~\Cref{fig:hardness distillation} for an illustration. Together, $D_1$ and $D_2$ form a coreset witnessing the fact that $f$ has decision tree complexity at least $s_1 + s_2$. To formalize this approach we rely on generalizations of two notions of boolean function complexity, namely certificate complexity and the relevance of variables, from the setting of total functions to partial functions. More formally, the two components of our method are as follows:  

\begin{enumerate} 
\item If there is an input $x^\star \in D_1$ such that the certificate complexity of $f$ on $x^\star$ relative to $D_1$ is at least $s_1$, then any decision tree $T$ that agrees with $f$ on $D_1$ must have a long path $\pi$ of length at least $s_1$. 
\item This path $\pi$ induces at least $s_1$ many subfunctions of $f$, corresponding to $f$ restricted by paths that diverge from $\pi$ at each of $\pi$'s at-least-$s_1$ many nodes.  If the number of variables of these subfunctions that are relevant relative to $D_2$ is at least $s_2$, then the at-least-$s_1$ many disjoint subtrees that branch off $\pi$ must have sizes that sum up to at least $s_2$. 
\end{enumerate}

%\gray{ We describe our technique for hardness distillation using \isedge\ as an illustrative example; the argument for the the amplified version is more elaborate is based on the same intuitions. 

% We begin with a brief sketch of the k+e lower bound in the No case of Theorem 2 (details are given in Section […]). Let T be any decision tree for IsEdge. We first observe that the “left spine” of T, the branch followed by the all-zeroes input, must have length at least k. This is because the vertices queried along this left spine must constitute a vertex cover: otherwise, the indicator $x[e]$ of at least one edge will also follow this left spine and receive the same classification as the all-zeroes input, which is a contradiction since IsEdge(x[e]) ne IsEdge(0). Next, we argue that the subtrees that are hanging off the left spine of T have sizes that sum to at least e, giving us an overall lower bound of k+e. To show this, we associate each edge $e = (v_i,v_j)$ with the first vertex along the left spine that queries one of its endpoints, say $v_i$. We then argue that the other endpoint, $v_j$, must be queried in the subtree $T_i$ that hangs off $v_i$; furthermore, all the edges associated with $T_i$ give rise to different queries.    }

\subsubsection{Hardness for constant error} 

To obtain hardness even against algorithms that are allowed constant error, we further improve the No case as follows: 
\medskip 

\begin{tcolorbox}[colback = white,arc=1mm, boxrule=0.25mm]
\vspace{3pt} 
\begin{center}{\bf The reduction for constant-error hardness}
\vspace{-4pt} 
\end{center} 
\begin{itemize}[leftmargin=0.5cm] 
\item[$\circ$] Yes case: If $G$ has a small vertex cover, then $f$ has small decision tree complexity.
\item[$\circ$] No case: If $G$ requires a large vertex cover, then there is a set $D\sse \zo^n$, a distributuon $\mathcal{D}$ over $D$, and a constant $\eps >0$ such that any decision tree that agrees with $f$ with probability $\ge 1-\eps$ over $\bx \sim \mathcal{D}$ must be large. 
\end{itemize} 
\end{tcolorbox}
\medskip

The key new ingredient in this final reduction is the hardness of $\alpha$-{\sl  Partial}\,{\sc VertexCover}, a relaxed version of {\sc VertexCover} where the goal is to find a set of vertices that cover a $1-\alpha$ fraction of vertices. We show that $\alpha$-{\sc PartialVertexCover} inherits its hardness of approximation from {\sc VertexCover} itself: 

\begin{claim}[Hardness of $\alpha$-{\sc PartialVertexCover}]
\label{claim:partial-vc-intro}
There are constants $\alpha\in (0,1)$ and $\delta>0$ such that if $\alpha$-\textsc{PartialVertexCover} on constant-degree, $n$-vertex graphs can be approximated to within a factor of $1+\delta$ in time $t(n)$, then ${\textsc{SAT}}$ can be solved in time $t(n\cdot\polylog(n))$. 
\end{claim}

This is thanks to the fact that {\sc VertexCover} is hard the approximate even for constant-degree graphs, which in turn follows from the PCP Theorem. 

With~\Cref{claim:partial-vc-intro} in hand,~\Cref{thm:main-intro} then follows by appropriately robustifying the other machinery described in this section.

\section{Discussion and future work}

Assuming SAT requires exponential time, \Cref{thm:main-intro} shows that the inherent time complexity of properly learning decision trees with queries is also exponential: the simple dynamic-programming-based Occam algorithm is essentially optimal, despite evidence to the contrary in the form of fast algorithms for various relaxations of the problem.  

\violet{A concrete problem left open by our work is that of optimizing the efficiency of our reduction, which takes an instance of SAT over $n$ variables and produces an instance of properly learning decision trees over $\tilde{O}(n^2)$ variables. Can this be improved to linear or quasilinear in $n$?} 

More broadly, a natural next step is to understand the complexity of {\sl weakly-proper} learning. As mentioned in the introduction, the landscape changes dramatically for this easier setting, and we have known since the 1980s of an algorithm that runs in quasipolynomial time~\cite{EH89}. This algorithm of Ehrenfeucht and Haussler has resisted improvement for over three decades and it is reasonable to conjecture that it is in fact optimal, even for query learners: 

\begin{conjecture}
\label{conj:optimality of EH} There is no algorithm that, given queries to a size-$s$ decision tree target and access to random labeled examples, runs in time $n^{o(\log s)}$ and returns an accurate decision tree hypothesis---one of any size, not necessarily $s$.
\end{conjecture}  

\Cref{table} places \Cref{thm:main-intro} and~\Cref{conj:optimality of EH} within the context of prior work:

\backrefsetup{disable}
\begin{table}[H]
  \captionsetup{width=.9\linewidth}
\begin{adjustwidth}{-5.5em}{}
\renewcommand{\arraystretch}{2}
\centering
\begin{tabular}{|c|c|c|}
\hline 
 & {\bf Random Examples} & {\bf Queries} \\ 
\hline 
\centering 
\begin{tabular}{c}
{\bf Proper}  \vspace{-12pt} \\
{\bf Learning} 
\end{tabular} & 
\begin{tabular}{c}
\cite{Ang,PV88}: Exponential lower bound. \vspace{-10pt} \\
Assumption: SAT requires exponential time \vspace{3pt}
\end{tabular} 
& 
\begin{tabular}{c}
\Cref{thm:main-intro}: Exponential lower bound. \vspace{-10pt} \\
Assumption: SAT requires exponential time \vspace{3pt} 
\end{tabular}
 \\ \hline 
%\begin{adjustwidth}{-2em}
\begin{tabular}{c}
{\bf Weakly-proper} \vspace{-10pt} \\
{\bf Learning} 
\end{tabular}
%\end{adjustwidth} 
& 
\begin{tabular}{c} 
\cite{ABFKP09,KST23}: Quasipoly lower bound. \vspace{-10pt} \\
Assumption: Inapproximability of  \vspace{-12pt} \\ parameterized {\sc SetCover} \vspace{3pt}
\end{tabular} 
& \Cref{conj:optimality of EH}: Quasipolynomial lower bound.  \\ \hline 
\end{tabular} 
\end{adjustwidth} 
\caption{Lower bounds for proper and weakly-proper learning of decision trees. In terms of upper bounds, the fastest known proper algorithm (dynamic-programming-based Occam algorithm) runs in exponential time, and the fastest known weakly-proper (Ehreunfeucht--Haussler) runs in quasipolynomial time.} 
\label{table} 
\end{table} 
\backrefsetup{enable}

Weakly-proper learning algorithms are akin to approximation algorithms, and the hardness of weakly-proper learning is akin to the hardness of approximation. An immediate, but not necessarily insurmountable obstacle in extending our techniques to the setting of weakly-proper learning is the fact that {\sc VertexCover}, whose hardness of approximation we rely on in our proof, is not {\sl that} hard to approximate: a simple greedy algorithm achieves a $2$-approximation. 

There is also more to be understood for (strongly-)proper learning of decision trees. Our work taken together with the recent query learner of~\cite{BLQT22} highlights, quite dramatically, the effect of distributional assumptions on the problem: our work gives an exponential lower bound in the distribution-free setting, whereas~\cite{BLQT22} gives an almost-polynomial time query algorithm for the uniform distribution. In the spirit of beyond worst-case analysis, an ambitious direction for future work is to understand the tractability of the problem vis-à-vis the complexity of the underlying distribution. An ultimate goal is to design efficient algorithms that circumvent the lower bounds established in this work, but nonetheless enjoy performance guarantees for the broadest possible class of distributions.

\violet{Finally, we believe that the notions of hardness distillation and coresets introduced in this work merit further study and could lead to more connections between the hardness of minimization problems and the hardness of learning.}

%% file: preliminaries.tex
\section{Preliminaries} 

\paragraph{Notation and naming conventions.}{We write $[n]$ to denote the set $\{1,2,\ldots,n\}$. We use lower case letters to denote bitstrings e.g. $x,y\in\zo^n$ and subscripts to denote bit indices: $x_i$ for $i\in [n]$ is the $i$th index of $x$. The string $x^{\oplus i}$ is $x$ with its $i$th bit flipped. We use superscripts to denote multiple bitstrings of the same dimension, e.g. $x^{(1)},x^{(2)},...,x^{(j)}\in\zo^n$. For a finite set $S$, $\text{Perm}(S)$ denotes the set of permutations of $S$. If $S=\{s_1,\ldots,s_{|S|}\}$, we identify $\pi\in\text{Perm}(S)$ with the tuple $(s_{i_1},\ldots, s_{i_{|S|}})$ where $\pi(s_j)=s_{i_j}$. In this setting, we simply write $\pi(j)$ to denote the $j$th element of the tuple, $\pi(j)=s_{i_j}.$
}

\paragraph{Distributions.}{We use boldface letters e.g. $\bx,\by$ to denote random variables. For a distribution $\mathcal{D}$, we write $\dist_{\mathcal{D}}(f,g)=\Pr_{\bx\sim \mathcal{D}}[f(\bx)\neq g(\bx)]$. A function $f$ is $\eps$-close to $g$ if $\dist_{\mathcal{D}}(f,g)\le \eps$. Similarly, $f$ is $\eps$-far from $g$ if $\dist_{\mathcal{D}}(f,g)>\eps$. The support of the distribution is the set of elements with nonzero mass and is denoted $\supp(\mathcal{D})$. 
}

\paragraph{{Decision trees.}}{ The size of a decision tree $T$ is its number of internal nodes and is denoted $|T|$. Two subtrees of $T$ are \textit{disjoint} if they do not share any internal nodes. In an abuse of notation, we also write $T$ for the function computed by the decision tree $T$. We say $T$ computes a function $f:\zo^n\to\zo$ if $T(x)=f(x)$ for all $x\in\zo^n$. The decision tree complexity of a function $f$ is the size of the smallest decision tree computing $f$ and is denoted $\dtsize(T)$.}

\paragraph{Restrictions and decision tree paths.}{
A restriction $\rho$ is a set $\rho\sse \{x_1,\overline{x}_1,\ldots,x_n,\overline{x}_n\}$ of literals, and $f_\rho$ is the subfunction obtained by restricting $f$ according to $\rho$: $f_\rho(x^\star)=f(x^\star\vert_\rho)$ where $x^\star\vert_\rho$ is the string obtained from $x^\star$ by setting its $i$th coordinate to $1$ if $x_i\in \rho$, $0$ if $\overline{x}_i\in \rho$, and otherwise setting it to $x^\star_i$. We say an input $x^\star$ is consistent with $\rho$ if $x_i\in \rho$ implies $x^\star_i=1$ and $\overline{x}_i\in \rho$ implies $x^\star_i=0$. 

We identify a depth-$d$, non-terminal path $\pi$ in a decision tree with a tuple of $d$ literals: $\pi=(\ell_1,\ell_2,\ldots,\ell_d)$ where each $\ell_i$ corresponds to a query of an input variable and is unnegated if $\pi$ follows the right branch and negated if $\pi$ follows the left branch. Paths naturally correspond to restrictions by forgetting their ordering. Therefore, we also write $f_\pi$ to denote the restriction of $f$ by $\{\ell_1,\ell_2,\ldots,\ell_d\}$. 
}

\paragraph{Graphs.}{An undirected graph $G=(V,E)$ has $n$ vertices $V\sse [n]$ and $m=|E|$ edges $E\sse V\times V$. The degree of a vertex $v\in V$ is the number of edges containing it: $|\{e\in E: v\in e\}|$. The graph $G$ is degree-$d$ if every vertex $v\in V$ has degree at most $d$. We often use letters $v,u,w$ to denote vertices of a graph $G$. 
}

\paragraph{Learning.}{In the PAC learning model, there is an unknown distribution $\mathcal{D}$ and some unknown \textit{target} function $f\in\mathcal{C}$ from a fixed \textit{concept} class $\mathcal{C}$ of functions over a fixed domain. An algorithm for learning $\mathcal{C}$ over $\mathcal{D}$ takes as input an error parameter $\eps\in (0,1)$ and has oracle access to an \textit{example oracle} $\textnormal{EX}(f,\mathcal{D})$. The algorithm can query the example oracle to receive a pair $(\bx,f(\bx))$ where $\bx\sim\mathcal{D}$ is drawn independently at random. The goal is to output a \textit{hypothesis} $h$ such that $\dist_{\mathcal{D}}(f,h)\le \eps$. Since the example oracle is inherently randomized, any learning algorithm is necessarily randomized. So we require the learner to succeed with some fixed probability e.g.~$2/3$. A learning algorithm is \textit{proper} if it always outputs a hypothesis $h\in\mathcal{C}$. A learning algorithm with \textit{queries} is given oracle access to the target function $f$ along with the example oracle $\textnormal{EX}(f,\mathcal{D})$. 

In this work we focus on the task of properly learning the concept class $\mathcal{T}_s =\{T:\zo^n\to\zo\mid T\text{ is a size-$s$ decision tree}\}$.
\begin{definition}[Properly PAC learning decision trees with queries]
An algorithm $\mathcal{L}$ properly learns $\mathcal{T}_s$ in time $t(n,s,\eps)$ if for all distributions $\mathcal{D}$ and for all $T\in\mathcal{T}_s$ and $\eps\in(0,1)$, $\mathcal{L}$ with oracle access to $\textnormal{EX}(T,\mathcal{D})$ and queries to $T$ runs in time $t(n,s,\eps)$ and, with probability $2/3$, outputs $h\in \mathcal{T}_s$ such that $\dist_{\mathcal{D}}(T,h)\le \eps$. 
\end{definition}
}

\paragraph{PCPs and {\sc Max-3Sat}.}{
For this work, we are interested in reductions from \SAT. Our techniques will rely on the hardness of approximation and  we therefore need a reduction from \SAT\ to approximating {\sc Max-3Sat}. The most efficient reduction exploits quasilinear PCPs: 

\begin{theorem}[Hardness of approximating {\sc Max-3Sat} via quasilinear PCPs \cite{Din07,EM08}]
\label{thm:quasilinear pcps}
    There is a constant $c\in (0,1)$ and a polynomial-time reduction that takes a 3CNF formula $\varphi$ with $m$ clauses and produces a 3CNF formula $\varphi^\star$ with $O(m\cdot\polylog(m))$ clauses satisfying
    \begin{itemize}
        \item if $\varphi$ is satisfiable then $\varphi^\star$ is satisfiable;
        \item if $\varphi$ is unsatisfiable then no assignment satisfies a $c$-fraction of clauses of $\varphi^\star$.
    \end{itemize}
\end{theorem}
}

\subsection{Hardness of Vertex Cover} 

\paragraph{Vertex cover.}{A vertex cover for an undirected graph $G=(V,E)$ is a subset of the vertices $C\sse V$ such that every edge has at least one endpoint in $C$. We write $\VC(G)\in \N$ to denote the size of the smallest vertex cover. See \Cref{fig:vertex cover prelims} for an example of a vertex cover. The {\vcover} problem is to decide whether a graph contains a vertex cover of size-$k$, i.e. to decide if $\VC(G)\le k$. We consider the more general \textit{gapped} vertex problem where the problem is to decide whether a graph has a small vertex cover or requires large vertex cover. Specifically we write $(k,k')$\textsc{-VertexCover} for the problem of deciding whether a graph contains a vertex cover of size-$k$ or every vertex cover has size at least $k'$. This gapped problem is equivalent to the problem of \textit{approximating} vertex cover. There is a polynomial-time greedy algorithm for vertex cover that approximates it within a factor of $2$, i.e.~solves $(k,2k)$\textsc{-VertexCover} in polynomial-time.

Constant factor hardness of \textsc{VertexCover} is known, even for bounded degree graphs (graphs whose degree is bounded by some universal constant). Papadimitriou and Yannakakis in \cite{PY91} give an approximation preserving reduction from {\sc Max-3Sat} to \textsc{VertexCover} on constant-degree graphs. The PCP theorem \cite{AS98, ALMSS98} implies $\NP$-hardness of approximating {\sc Max-3Sat} and therefore, combined with the reduction in \cite{PY91}, implies hardness of approximating \textsc{VertexCover} on constant-degree graphs. (For a further discussion and history of these results, see the survey by Trevisan \cite{Tre14}.)

\begin{theorem}[Hardness of approximating \textsc{VertexCover}]
\label{thm:hardness of vertex cover} 
There are constants $\delta>0$ and $d\in \N$ such that if $(k,(1+\delta)\cdot k)$\textsc{-VertexCover} on $n$-vertex degree-$d$ graphs can be solved in time $t(n)$, then \SAT\ can be solved in time $t(n\cdot\polylog(n))$.
\end{theorem}
}

This hardness follows from \Cref{thm:quasilinear pcps} and the reduction in \cite{PY91}. The $n\cdot\polylog(n)$ factor originates from \Cref{thm:quasilinear pcps}. 

The fact that \Cref{thm:hardness of vertex cover} holds for constant degree graphs will be essential for our lower bound because it allows us to assume that $k$ is large: $\VC(G)=\Theta(m)$.

\begin{fact}[Constant degree graphs require large vertex covers]
    \label{fact:constant degree graphs have large vc size}
    If $G$ is an $m$-edge degree-$d$ graph, then $\VC(G)\ge m/d$.
\end{fact}

This fact follows from the observation that in a degree-$d$ graph each vertex can cover at most $d$ edges. 

\input{figvcprelims}

%% file: figvcprelims.tex
\begin{figure}[h!]
    \centering
    \begin{tikzpicture}[
      % every node/.style={draw,circle,minimum size=0.25cm,inner sep=-0pt},
      % usnode/.style={fill=black},
      % ssnode/.style={fill=black},
      % every fit/.style={ellipse,draw,inner sep=-2pt,text width=2cm},
      % ->,shorten >= 2pt,shorten <= 2pt
    ]
        \centering
        \node[shape=circle,draw=black,inner sep=1.5pt,fill=black] (WW) at (-1,0) {};
        \node[shape=circle,draw=black,inner sep=1.5pt,fill=black] (WN) at (-0.71,0.71) {};
        \node[shape=circle,draw=black,inner sep=1.5pt,fill=black] (WS) at (-0.71,-0.71) {};
        \node[shape=circle,draw=teal,inner sep=1.5pt,fill=teal] (W) at (0,0) {};
        
        \node[shape=circle,draw=teal,inner sep=1.5pt,fill=teal] (C) at (1,0) {};
        \node[shape=circle,draw=black,inner sep=1.5pt,fill=black] (CN) at (1,1) {};
        \node[shape=circle,draw=black,inner sep=1.5pt,fill=black] (CS) at (1,-1) {};

        \node[shape=circle,draw=teal,inner sep=1.5pt,fill=teal] (E) at (2,0) {};
        \node[shape=circle,draw=black,inner sep=1.5pt,fill=black] (EE) at (3,0) {};
        \node[shape=circle,draw=black,inner sep=1.5pt,fill=black] (EN) at (2.71,0.71) {};
        \node[shape=circle,draw=black,inner sep=1.5pt,fill=black] (ES) at (2.71,-0.71) {};
        
        \path [-] (WW) edge node[left] {} (W);
        \path [-] (WN) edge node[left] {} (W);
        \path [-] (WS) edge node[left] {} (W);
        \path [-] (C) edge node[left] {} (W);
        
        \path [-] (CN) edge node[left] {} (C);
        \path [-] (CS) edge node[left] {} (C);
        \path [-] (E) edge node[left] {} (C);
        
        \path [-] (EE) edge node[left] {} (E);
        \path [-] (EN) edge node[left] {} (E);
        \path [-] (ES) edge node[left] {} (E);
    \end{tikzpicture}
  \captionsetup{width=.9\linewidth}
    \caption{A graph $G=(V,E)$ with $10$ edges having $\VC(G)=3$. The unique vertex cover of size $3$ is highlighted in teal.}
    \label{fig:vertex cover prelims}
\end{figure}
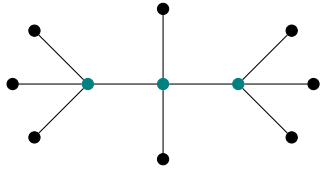

%% file: main_reduction.tex
\section{A reduction from $\mathrm{\sc VertexCover}$ to Decision Tree Minimization}

\subsection{Intuition and warmup: the ${\mathrm{\isedge}}_G$ function} 

In this section we prove~\Cref{claim:isedge-intro}, which serves as a warmup for our core reduction,~\Cref{thm:ell-isedge-intro}.  We first introduce a few notions (and notation) that will be useful throughout the rest of the paper. 

\subsubsection{Useful notions and notation: edge partitions and divergent path prefixes} 

\paragraph{Edge partitions induced by decision trees for $\mathrm{\isedge}$.}{
We will make use of the notion of a \textit{restricted edge neighborhood} and a \textit{restricted vertex neighborhood}. Specifically, we will be interested in the edges incident to a particular vertex which do \textit{not} contain certain vertices.

\begin{definition}[Restricted edge and vertex neighborhood]
    \label{defn:vertex neighborhood}
    For a graph $G=(V,E)$, the edge neighborhood of $v_{i_\kappa}\in V$ restricted by $v_{i_1},\ldots,v_{i_{\kappa-1}}$, denoted $E(v_{i_\kappa}; v_{i_1},\ldots,v_{i_{\kappa-1}})$, is the set of edges containing $v_{i_\kappa}$ but \textit{not} any of $v_{i_1},\ldots,v_{i_{\kappa-1}}$: 
    $$
    E(v_{i_\kappa}; v_{i_1},\ldots,v_{i_{\kappa-1}})\coloneqq \{e\in E\mid v_{i_\kappa}\in e\text{ and }v_{i_1},\ldots,v_{i_{\kappa-1}}\not\in e\}.
    $$
    
    The vertex neighborhood of $v_{i_\kappa}$ restricted by $v_{i_1},\ldots,v_{i_{\kappa-1}}$, denoted $V(v_{i_\kappa}; v_{i_1},\ldots,v_{i_{\kappa-1}})$, is the set of neighbors of $v_{i_\kappa}$ excluding the vertices $v_{i_1},\ldots,v_{i_{\kappa-1}}$:
    $$
    V(v_{i_\kappa}; v_{i_1},\ldots,v_{i_{\kappa-1}})\coloneqq \left\{v\in V\mid \{v_{i_\kappa},v\}\in E\text{ and } v\neq v_{i_1},\ldots,v_{i_{\kappa-1}} \right\}.
    $$
\end{definition}

Often when a tuple of vertices $(v_{i_1},\ldots,v_{i_k})$ is understood from context, we will use the shorthand notation $E_\kappa=E(v_{i_\kappa};v_{i_1},\ldots,v_{i_{\kappa-1}})$ for $\kappa=1,\ldots,k$ and likewise for $V_\kappa$. Restricted edge and vertex neighborhoods are closely related to each other, and each can be defined in terms of the other:
$$
E_\kappa = \left\{\{v_{i_\kappa},v\}\mid v\in V_\kappa\right\}\quad\text{and}\quad V_\kappa=\{v\mid \{v_{i_\kappa},v\}\in E_\kappa\}.
$$

Given a vertex cover $\{v_{i_1},\ldots,v_{i_k}\}$, the sets $\{E_\kappa\}_{\kappa\in [k]}$ form a partition of the edge set $E$. Indeed,
$$
\bigcup_{\kappa\in [k]} E_\kappa = E
$$
since every edge in $G$ is incident to some vertex $v_{i_\kappa}$. Also, the sets $E_\kappa$ are disjoint since each $E_\kappa$ excludes the edges already covered by the previous $E_1,\ldots,E_{\kappa-1}$ sets. In fact, the converse also holds. If $C=\{v_{i_1},\ldots,v_{i_k}\}$ are vertices such that $E_\kappa$ partition the edge set then $C$ must form a vertex cover: every edge $e\in E$ is in some partition $E_\kappa$ and so $v_{i_\kappa}$ covers $e$.

\begin{fact}
\label{fact:restricted vertex neighborhood partitions}
Let $C=\{v_{i_1},\ldots,v_{i_k}\}$ be a subset of vertices of a graph $G$ and $E_\kappa\coloneqq E(v_{i_\kappa};v_{i_1},\ldots,v_{i_{\kappa-1}})$ for $\kappa\in [k]$. Then $C$ forms a vertex cover of $G$ if and only if $\{E_\kappa\}_{\kappa\in[k]}$ form a partition of $E$.  
\end{fact}

A key property of the $\mathrm{\isedge}_G$ function is that every decision tree for it induces such an edge partition in the following way. Every decision tree for $\mathrm{\isedge}_G$ has a path $\pi$ in it whose path variables form a vertex cover. This vertex cover induces a partition of the edges of $G$. Each part of the partition corresponds to a unique variable in this decision tree path. This correspondence will be important for lower bounding the size of the decision tree in the case when $G$ requires large vertex covers. To describe this correspondence, it will be useful for us to have the following notation for a path that diverges from from $\pi$ at a particular point and then stops.

\begin{definition}[Divergent path prefix; see~\Cref{fig:divergent path prefix}]
    \label{defn:path prefix}
    For a path $\pi$, the path $\pi\vert_{\oplus \kappa}$ denotes the path which follows $\pi$ for the first $\kappa-1$ queries, flips the $\kappa$th query, then terminates:
    $$
    \pi\vert_{\oplus \kappa}\coloneqq \left(\pi(1),\ldots, \pi(\kappa-1),\overline{\pi(\kappa)}\right).
    $$
\end{definition}

\input{figdivergentpathprefix}

If $\pi$ is the path corresponding to a vertex cover, then $\pi\vert_{\oplus \kappa}$ corresponds to the path followed by edges in $E_\kappa$ (here we are conflating edges and edge indicator strings).
}

\subsubsection{Proof of~\Cref{claim:isedge-intro}}

\begin{proof}[Proof of the Yes case] 
Let $C = \{ v_{i_1},\ldots, v_{i_k}\}$ be a vertex cover for $G$. The leftmost branch $\pi$ of our decision tree queries these vertices successively and terminates with a $0$-leaf. These are the vertices colored blue in~\Cref{fig:isedge upper bound}.

We move on to describing each of the subtrees branching off of $\pi$. More formally, for each $\kappa\in [k]$, we describe the subtree rooted at the end of $\pi |_{\oplus \kappa}$ (i.e.~the subtree that is the $1$-successor of $v_{i_k}$). %Recall the definition $E_\kappa$ and $V_\kappa$ and note that every $\ind[e]$ with $e\in E_\kappa$ follows the path $\pi |_{\oplus \kappa}$. 
At this point $T$ ``knows" that $v_{i_k}$ is set to $1$. For $\mathrm{\isedge}$ to output $1$, exactly one of $v_{i_\kappa}$'s neighbors must also be set to $1$, and all $n-2$ other vertices must be set to $0$ (i.e.~these are precisely the inputs $\ind[e]$ for $e\in E_{\kappa}$). Therefore $T$ queries all $v\in V_\kappa$ (i.e.~the neighbors of $v_{i_\kappa}$ that have not already been queried along $\pi$), testing to see whether any of them are 1, and terminates with a 0-leaf if they are all set to $0$. These are the vertices colored teal in~\Cref{fig:isedge upper bound}.

Finally, we describe the subtree that is the $1$-successor of each $v\in V_\kappa$. At this point $T$ knows that $v_{i_{\kappa}}$ and this neighbor $v$ are both set to $1$, and it remains only to check that all other vertices are set to $0$ before outputting $1$: it queries all $n-2$ vertices in $V$ that are not $v$ or $v_{i_\kappa}$ and outputs $1$ iff all of them are set to $0$. These are the vertices colored orange in~\Cref{fig:isedge upper bound}. 

We complete the proof by bounding the size of $T$. Its leftmost branch has size $k$ (the blue vertices). By~\Cref{fact:restricted vertex neighborhood partitions}, querying all $v\in V_\kappa$ for $\kappa\in [k]$ results in an additional $\sum_\kappa |V_\kappa| =  \sum_\kappa |E_\kappa|= m$ internal nodes (the teal vertices). After each of these $m$ internal nodes, we query $n-2$ more vertices, resulting in an additional $m(n-2) < mn$ internal nodes (the orange vertices).  Thus, the total size of $T$ is at most $k+m+mn$.
\end{proof} 

\input{figisedgeupperboundv2}

We proceed to a proof of the lower bound.

\begin{proof}[Proof of the No case] 
Our proof consists of two parts: (1) proving that the leftmost branch of $T$ must be a vertex cover and therefore has size at least $k'$ and (2) showing that the rest of the tree has size at least $m$. See~\Cref{fig:isedge lower bound} for an illustration. 

\begin{enumerate} 
\item {\sl Leftmost branch must be a vertex cover.} Let $\pi$ be the leftmost branch of $T$  and suppose for contradiction that the vertices queried along $\pi$ do not form a vertex cover for $G$. This means that there is some edge $e\in E$ that is not queried along $\pi$, and hence both $\ind[e]$ and $0^n$ will follow $\pi$ and reach the same leaf. Since $\mathrm{\isedge}(0^n) = 0 \ne 1 = \mathrm{\isedge}(\ind[e])$, this is a contradiction. 

\item {\sl Rest of the tree has at least $m$ nodes.}  Let us order the vertices of $\pi$ from root downwards as $v_{i_1}, \ldots , v_{i_{|\pi|}}$. For each $\kappa\in [|\pi|]$, we consider the subtree $T_\kappa$ that is the $1$-successor of $v_\kappa$. %This subtree must correctly compute the restricted function \isedge$_{\pi|_{\oplus \kappa}}$, and in particular, it must correctly compute all $\ind[e]$ for $e\in E_\kappa$. 
Consider $e \in E_{\kappa}$ and suppose $e=(v_{i_\kappa},v)$. By the definition of $E_\kappa$, the endpoint $v$ has not yet been queried when $\ind[e]$ enters $T_\kappa$. Thus, $T_\kappa$ must query $v$, since otherwise $T$ cannot distinguish between $\ind[e]$ and $\ind[e]^{\oplus v}$ (note that $\mathrm{\isedge}(\ind[e]) = 1 \ne 0 = \mathrm{\isedge}(\ind[e]^{\oplus v})$). Further, all $e\in E_\kappa$ will have distinct second endpoints that $T_\kappa$ must query (since if not, then they would share both their endpoints and be the exact same edge). In other words, we have argued that $T_\kappa$ must query all the vertices in $V_\kappa$. %We conclude that this subtree must query $|E_\kappa|$ additional vertices. 

Since the sets $E_\kappa$ for $\kappa\in[|\pi|]$ partition the edges (\Cref{fact:restricted vertex neighborhood partitions}), we have that all these disjoint subtrees $T_1,\ldots,T_{|\pi|}$ taken together must query at least $\sum_\kappa |V_\kappa| = |E_\kappa| = |E|= m$ additional vertices.  
\end{enumerate} 
Combining the two claims above we show shown that $|T| \ge k'+m$ and the proof is complete. \end{proof}

\input{figisedgelowerbound}

%In summary, if $G$ has a small vertex cover, then \isedge$_G$ has a size $|T|$ decision tree, and if every vertex cover for $G$ is large, then every decision tree for \isedge$_G$ has size at least $|T'|$. Crucially,  $|T|>|T'|$, which is the critical failure of the above construction. In the next section, we provide a modified function for which we show $|T|<|T'|$. 

\subsection{$\ell$-$\mathrm{\isedge}$: an amplified version of {\isedge}}
\label{sec:ell-isedge}

%\lnote{Some prose here, emphasize that we are only amplifying the $1$-sensitivity of satisfying assignments}

%\gray{ 
%\begin{definition}[The $\ell$-amplified {\isedge} function]\lnote{I find this definition slightly hard to parse} 
 %   For an $n$-vertex graph $G$, the $\ell$-amplified {\isedge} function with respect to $G$, $\ell\text{-\isedge}_G:\zo^N\to\zo$ on $N\coloneqq n+\ell n$ variables is defined as
    % on $N\coloneqq n+\ell n$ variables is defined as
    % \begin{align*}
    %     &\ell\text{-\isedge}:\zo^{n}\times (\zo^\ell)^n\to\zo\\
    %     &(v,v^{(1)},\ldots,v^{(n)})\mapsto  \begin{cases}
    % 1 & \mathrm{\isedge}(v)=1\text{ and }v_i=1\Rightarrow \mathrm{AND}(v^{(i)})=1\\
    % 0 & \mathrm{otherwise}
    % \end{cases}
    % \end{align*}

 %   $$
 %   \ell\text{-\isedge}_G(v^{(0)},v^{(1)},\ldots,v^{(n)})
 %   =
 %   \begin{cases}
%    1 & \mathrm{\isedge}_G(v^{(0)})=1\text{ and }v_i^{(0)}=1\Rightarrow \mathrm{AND}(v^{(i)})=1\\
%    0 & \mathrm{otherwise}
 %   \end{cases}.
%    $$
%\end{definition}} 

\begin{definition}[The $\ell$-amplified {\isedge} function]\label{def:ell-isedge} Let $G = (V,E)$ be an $n$-vertex graph and $\ell \in \N$. The {\sl $\ell$-amplified edge indicator function of $G$} is the function 
\[ \ell\text{-}\mathrm{\isedge}_G : \zo^n \times (\zo^n)^\ell \to \zo\] 
 defined as follows: $\ell\text{-}\mathrm{\isedge}_G(v^{(0)},v^{(1)},\ldots,v^{(\ell)}) = 1$ iff 
\begin{enumerate} 
\item $\mathrm{\isedge}_G(v^{(0)}) = 1$ (i.e.~$v^{(0)} = \Ind[e]$ for some $e\in E$), and
\item $v_{i}^{(1)} = \cdots = v_i^{(\ell)} = 1$ for all $i\in [n]$ such that $v_i^{(0)} = 1$. 
\end{enumerate} 
\end{definition}

\paragraph{Notation and terminology.} When $G$ is clear from context, we drop the subscript and just write $\ell$-{\isedge}. We also use $N\coloneqq n+n\ell$ to denote the number of inputs to $\ell$-{\isedge}.  We refer to $v^{(0)}_1,\ldots,v^{(0)}_n$ as the {\sl original} variables. As in the nonamplified ${\mathrm{\isedge}}$ function, there is a natural correspondence between these original variables and the vertices $V = \{v_1,\ldots,v_n\}$ of $G$. For each original variable $v^{(0)}_i$, we refer to $v^{(1)}_i,\ldots,v^{(\ell)}_i$ as its {\sl duplicated} variables and write
\[ \textsc{Dup}(v_i)\coloneqq \left\{v_i^{(1)},\ldots,v_i^{(\ell)}\right\}. \]
We write $\ell\text{-}\ind[e] \in (\zo^n)^{\ell+1}$ to denote the string $(\ind[e],\ldots,\ind[e])$. Note that $\ell\text{-}\mathrm{\isedge}(\ell\text{-}\ind[e]) = 1$ and these are the $1$-inputs of minimum Hamming weight. 

\paragraph{Asymmetries in the definition of $\ell\text{-}\mathrm{\isedge}$.} We note two sources of asymmetry in the definition of $\ell$-$\mathrm{\isedge}$, both of which are crucial for~\Cref{thm:ell-isedge-intro} (specifically,~\Cref{rem:asymm}) to hold. First, the original variables play a distinct role from the duplicated ones: for $\ell$-${\mathrm{\isedge}}$ to output $1$, the original variables have to agree with an edge indicator but the duplicated variables do not.  Second, there is also an asymmetry between $1$- and $0$-coordinates: for $\ell$-$\mathrm{\isedge}$ to output $1$, the duplicated variables have to be set to $1$ whenever the original variables are set to $1$, but the same is not true for the $0$-coordinates.

\subsubsection{Proof of~\Cref{thm:ell-isedge-intro}}
\input{figellisedgeupperbound}
\begin{proof}[Proof of the Yes case]     The construction is a slight extension of our tree for $\mathrm{\isedge}$ that we constructed for the Yes case of~\Cref{claim:isedge-intro}. See \Cref{fig:ell isedge upper bound} for an illustration of this construction. Let $C = \{v_{i_1},\ldots,v_{i_k}\}$ be a vertex cover of $G$. Similar to before, the leftmost branch $\pi$ of our tree $T$ queries the original variables $v^{(0)}_{i_1},\ldots,v^{(0)}_{i_k}$ corresponding to these vertices and terminates with a $0$-leaf.

    We now describe the subtree $T_\kappa$ that is the $1$-successor of $v^{(0)}_{i_\kappa}$ for $\kappa\in [k]$. It first checks if the duplicated variables of $v_{i_\kappa}$ are all set to $1$, since that is a necessary criterion for $\ell\text{-}\mathrm{\isedge}$ to output $1$: it queries all $\ell$ variables in $\textsc{Dup}(v_{i_\kappa})$ and outputs $0$ once any of them are set to $0$.  If all of them are indeed set to $1$, then for $T_\kappa$ to output $1$ it must check that (i) there is a neighbor $v$ of $v_{i_k}$ whose original variable is set to $1$, (ii) all the other original variables are set to $0$, and (iii) the duplicated variables of $v$ are set to $1$. 
    
    For (i) and (ii), $T_\kappa$ queries the remaining original variables in a manner identical to the tree for $\mathrm{\isedge}$. Briefly restating that construction, it verifies Condition (i) by querying the original variables of $v\in V_\kappa$, testing to see of any of them are set to 1. If none of them are set to $1$, it outputs $0$. Otherwise, once the original variable of some $v\in V_\kappa$ is set to $1$, the tree $T_\kappa$ moves on to verifying Condition (ii): it queries the original variables of all $n-2$ vertices in $V \setminus \{v_{i_\kappa}, v\}$ and outputs $0$ once any of them are set to $1$. If all of them are indeed set to $0$, it moves on to verifying Condition (iii). It queries all $\ell$ variables in $\textsc{Dup}(v)$ and outputs $1$ iff all of them are set to $1$. 

We now bound the size of this tree. The portion of it that is identical to the tree for $\mathrm{\isedge}$ will have size at most $k+m+mn$, as proved in~\Cref{claim:isedge-intro}. We incur an additional $k\ell$ nodes to query the duplicate variables for the vertex cover: $\ell$ duplicate variables for each vertex in the size-$k$ vertex cover. Finally, since we additionally query the $\ell$ many duplicate variables for each $v \in V_\kappa$ in the subtree $T_\kappa$, we incur another additional $\ell \sum_{\kappa} |V_\kappa| = \ell \sum_{\kappa} |E_\kappa| = \ell m$
many nodes.  In total, this results in a tree of size 
\[ |T| \le k+ m+mn + k\ell + \ell m = (\ell +1)(k+m)+mn. \qedhere \]  
\end{proof} 

We now prove the lower bound. %We use the following helpful notation. For $e\in E$, let $\ell-\ind[e]:= (\ind[e])^{\ell+1} $ be a generalized edge indicator for $\ell$-\isedge. \carnote{Maybe want to add this notation to the definition of $\ell$-\isedge} It sets all variables to 0 that do not correspond to an endpoint of $e$ (or duplicate of an endpoint).

\begin{proof}[Proof of No case]
Just as in the proof of~\Cref{claim:isedge-intro}, we divide our proof into two parts. We show that (1) the leftmost branch of any decision tree for $\ell$-\isedge\ must correpond to a vertex cover and hence has size at least $k'$, and (2) the rest of the tree must have size at least $\ell k'+(\ell+1)m$.

\begin{enumerate} 
\item {\sl Leftmost branch must be a vertex cover.} Let $\pi$ be the leftmost branch of $T$  and $v_{i_1}^{(j_1)},\ldots,v_{i_{|\pi|}}^{(j_{|\pi|})}$ be the variables queried along $\pi$.  We claim that the corresponding vertices $v_{i_1},\ldots,v_{i_{|\pi|}}\in V$ must form a vertex cover for $G$.  Suppose for contradiction that they do not. This means that there is some edge $e \in E$ such that neither the original nor duplicated variables of $e$'s endpoints are queried along $\pi$. Therefore both $\ell\text{-}\ind[e]$ and $0^N$ will follow $\pi$ and reach the same leaf. Since $\mathrm{\isedge}(0^N) = 0 \ne 1 = \ell\text{-}\mathrm{\isedge}(\ell\text{-}\ind[e])$, this is a contradiction. 

\item {\sl Rest of the tree has at least $\ell k' +(\ell +1)m$ nodes.}  Order the vertices of $\pi$ from root downwards as $v_{i_1}, \ldots , v_{i_{|\pi|}}$. We will consider only the indices $\kappa \in [|\pi|]$ such that $E_{\kappa}$ is nonempty, noting that the vertices corresponding to these indices still form a vertex cover, and hence there are at least $k'$ such indices. Fix such a $\kappa$ and consider the subtree $T_\kappa$ that is the $1$-successor of $v_{i_{\kappa}}^{(j_{\kappa})}$. Let $e = (v_{i_\kappa},v)\in E_\kappa$. We argue that $v_{i_{\kappa}}$ is responsible for $\ell$ additional queries within $T_\kappa$, and $v$ for $\ell+1$ additional ones.  For the former claim, let $j \in \{0,\ldots,\ell\}\setminus \{ j_\kappa\}$. By the definition of $E_\kappa$, the variable $v_{i_\kappa}^{(j)}$ has not yet been queried when $\ell\text{-}\ind[e]$ enters $T_\kappa$. Since 
\[ \ell\text{-}\mathrm{\isedge}(\ell\text{-}\ind[e]) = 1 \ne 0 = \ell\text{-}\mathrm{\isedge}(\ell\text{-}\ind[e]^{\oplus v_{i_\kappa}^{(j)}}), \]
it follows that $v^{(j)}_{i_\kappa}$ must be queried within $T_\kappa$. Similarly, the latter claim follows from the fact that $T_\kappa$ must query the original and all the duplicated variables of $v$, a total of  $\ell+1$ many variables. This latter claim holds for all endpoints of edges $e\in E_\kappa$ (i.e.~the vertices $v\in V_{\kappa}$), so in total $T_\kappa$ must contain at least $\ell + (\ell+1)|V_\kappa|$ nodes. 
Summing over all $\kappa \in [|\pi|]$ such that $E_\kappa$ is nonempty and applying~\Cref{fact:restricted vertex neighborhood partitions}, we get that the disjoint subtrees $T_1,\ldots,T_{|\pi|}$ must query at least 
\[ \sum_{\kappa\colon E_\kappa\ne \varnothing} \ell + (\ell+1)|E_\kappa| = \ell k' + (\ell+1)m \]
many variables.
\end{enumerate} 
Combining the two claims above we have shown that \[ |T| \ge k' + \ell k' + (\ell+1)m  = (\ell+1)(k'+m) \]  and the proof is complete.
\end{proof}

\subsection{Hardness of decision tree minimization} 
\medskip 

\begin{tcolorbox}[colback = white,arc=1mm, boxrule=0.25mm]
\vspace{3pt} 
\textsc{DT-Min}:  Given a decision tree $T^\star:\zo^n\to\zo$, construct a minimum-size decision tree $T$ such that $T \equiv T^\star$ (i.e.~$T(x)=T^\star(x)$ for all $x\in\zo^n$).
\vspace{3pt} 
\end{tcolorbox}
\medskip

This problem of decision tree minimization was first shown to be $\NP$-hard by Zantema and Bodlaender \cite{ZB00}. That result was subsequently improved by Sieling \cite{Sie08} who showed that the problem is even $\NP$-hard to approximate. Using \Cref{thm:ell-isedge-intro} we recover this hardness of approximation.  We begin by observing that our proofs of the Yes and No cases of~\Cref{thm:ell-isedge-intro} are algorithmic in the following sense:  
\begin{itemize} 
\item[$\circ$] In the Yes case, we showed that given a graph $G$ and a size-$k$ vertex cover, the tree $T$ for $\ell\text{-}\mathrm{\isedge}$ of size $(\ell+1)\cdot (k+m)+ mn$ can be constructed in $\poly(\ell,n)$ time. 
\item[$\circ$] In the No case, we showed that given a size-$s'$ tree $T$ for $\ell\text{-}\mathrm{\isedge}$, a size-$k'$ vertex cover for $G$ satisfying $(\ell+1)\cdot (k'+m) \le s'$ can be constructed in $\poly(\ell,n)$ time. 
\end{itemize} 

With these observations in hand, we are now ready to recover~\cite{Sie08}'s result. 

\begin{lemma}[A reduction from {\sc VertexCover} to {\sc DT-Min}]
\label{lem:hardness of apx dt size}
 There is a polynomial-time reduction that takes a degree-$d$, $n$-vertex, $m$-edge graph $G$ and  produces a decision tree $T^\star$ such that the following holds. Given any tree $T$ such that $T\equiv T^\star$ and whose size is within a $(1+\delta)$ factor of the optimal for $T^\star$, one can construct in polynomial time a size-$k'$ vertex cover of $G$ satisfying $k' \le (1+\delta')\cdot \VC(G)$ where $\delta' = O(d\delta)$. 
\end{lemma}

{ 
\begin{proof} 
Let $\ell \coloneqq 2mn$. We begin by applying the Yes case of~\Cref{thm:ell-isedge-intro} to $G$ with the trivial vertex cover of all~$n$ vertices to obtain a decision tree $T^\star$ for $\ell\text{-}\mathrm{\isedge}$ of size 
\[ (\ell + 1)\cdot (n+m) + mn = (2mn + 1) \cdot (n+m) + mn. \] 
As observed above, our proof of~\Cref{thm:ell-isedge-intro} shows that $T^\star$ can be constructed from $G$ in $\poly(n)$ time. This tree $T^\star$ will be the input to {\sc DT-Min} in our reduction. 

By the Yes case of~\Cref{thm:ell-isedge-intro} again, if $\VC(G) \eqqcolon k$ then \[ \dtsize(\ell\text{-}\mathrm{\isedge}) \le (\ell+1) (k+m) + mn \eqqcolon s. \] 
Suppose an algorithm for {\sc DT-Min} returns a tree $T$ for $\ell\text{-}\mathrm{\isedge}$ of size $s'$ where $s' \le (1+\delta)\cdot s$. We claim that we can then efficiently construct a vertex cover for $G$ of size $k'$ where $k' \le (1+\delta')\cdot k$ and $\delta' = O(d\delta)$, thereby completing the reduction. Our proof of~\Cref{thm:ell-isedge-intro} shows that we can efficiently construct from $T$, in $\poly(n)$ time, a size-$k'$ vertex cover satisfying: 
\[ (\ell+1)(k'+m) \le s'.  \] %\carnote{this should be an inequality right? $s' \geq (\ell+1)(k'+m)$ it doesn't affect the rest of the proof}
The claim that $s' \le (1+\delta)\cdot s$ is therefore equivalent to
\[  (\ell+1) \cdot (k'+m) \le (1+\delta) \cdot \big[(\ell+1) (k+m) + mn\big]. \] 
Rearranging the above, we get that 
\begin{align*} 
 k' &\le (1+\delta)\cdot k + \delta m + \frac{(1+\delta)\cdot mn}{\ell+1} \\ 
 &\le (1+\delta)\cdot k + \delta k d + \frac{(1+\delta)\cdot mn}{\ell+1} \tag*{\text{($m \le kd$ by~\Cref{fact:constant degree graphs have large vc size})}}\\ 
 &< (1+\delta)\cdot k + \delta k d + 1 \tag{Our choice of $\ell$} \\
 &< \big[  1 + \delta(d+2)\big] \cdot k  
\end{align*}%\carnote{confused how we got to first line. shouldn't it be $1+\delta -m $ instead of $\delta m$}
and the proof is complete. 
\end{proof} 
}

\cite{Sie08}'s result now follows as an immediate consequence of~\Cref{lem:hardness of apx dt size} and the fact that {\sc VertexCover} is hard to approximate even for constant-degree graphs (\Cref{thm:hardness of vertex cover}):  

\begin{theorem}[Hardness of approximation for {\sc DT-Min}~\cite{Sie08}]
There is a constant $\delta\in (0,1)$ such that if \textsc{DT-Min} can be approximated to within a factor of $1+\delta$ in polynomial-time, then $\Ptime=\NP$. 
\end{theorem}

(\cite{Sie08} then amplifies this constant-factor inapproximability to a superconstant factor  using an XOR lemma from~\cite{HJLT96}. We refer the interested reader to~\cite{Sie08} for the details of this step.)

In the next section, we strengthen \cite{Sie08}'s result by showing that the same hardness holds even if the algorithm need only minimize $T$ over a small set of input points rather than all of $\zo^n$.

%% file: figdivergentpathprefix.tex
\begin{figure}[h!]
    \centering
    \begin{tikzpicture}[tips=proper]
        \node[isosceles triangle,
            draw,
            thick,
            isosceles triangle apex angle=60,
            rotate=90,
            minimum size=6cm] (T1) at (0,0){};
        \draw[violet,dashed] (T1.east) .. controls ([xshift=-0.1cm]T1.358) .. ([yshift=-0.5cm]T1.east) node[] (N1) {{}};
        \draw[violet,dashed] ([yshift=-0.5cm]T1.east) .. controls ([xshift=0.4cm]T1.40) and (T1.350) .. (T1.center) node[fill=white,pos=0.8] (N1) {$\pi$};
        \draw[violet,dashed] (T1.center) .. controls ([xshift=2.4cm]T1.110) .. (T1.west) node[] (N2) {{}};

        % xi point and label
        \node[draw,circle,fill=black,inner sep=0.8pt] (x) at (T1.center) {};

        % divergent path
        \draw[-{Stealth[scale=0.75]}] (x) to ([xshift=0.4cm,yshift=-0.4cm]T1.center);
        \draw[] ([xshift=0.7cm,yshift=-0.1cm]T1.center) node [black] {$\violet{\pi}\vert_{\oplus \kappa}$};

        % depth kappa label
        \draw[] ([xshift=-.8cm,yshift=0cm]T1.center) node [black] {\small depth $\kappa$};
        % path label
        % \draw[] (x) to ([xshift=1.5cm]T1.right side);
        % \draw[] ([xshift=1.5cm]T1.right side) node [right,black,fill=white, text width=6cm] {For };

    \end{tikzpicture}
  \captionsetup{width=.9\linewidth}
    \caption{Illustration of a divergent path prefix. The root-to-leaf path $\pi$ is illustrated in purple. At depth $\kappa$ the path $\pi\vert_{\kappa}$ diverges and terminates.}
    \label{fig:divergent path prefix}
\end{figure}
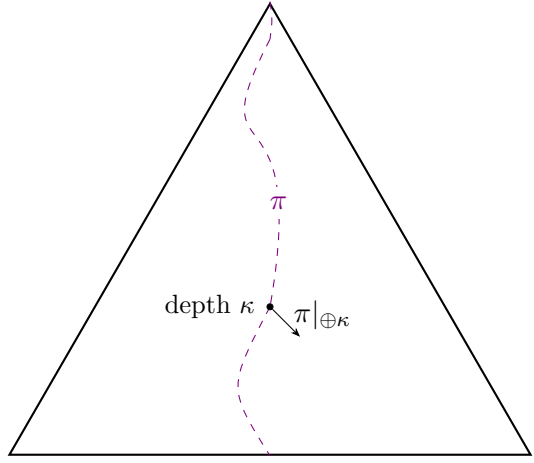

%% file: figisedgeupperboundv2.tex
\begin{figure}[H]
  \captionsetup{width=.9\linewidth}

    \centering
    \begin{tikzpicture}[tips=proper]
        % points of left path
        \draw[] (3,3) node [fill=white,text=blue] (root) {$v_{i_1}$};
        \draw[] (3.7, 3.05) node [fill=white] (rootlabel) {\footnotesize (root)};
        \draw[] (2.5,2.5) node [fill=white,rotate=-45] (incoming) {$\vdots$};
        \draw[] (2,2) node [fill=white,text=blue] (x1) {{$v_{i_\kappa}$}};
        \draw[] (0.5,0.5) node [fill=white,text=blue] (x2) {$v_{i_{\kappa+1}}$};
        \draw[] (-1,-1) node [fill=white,text=blue] (x3) {$v_{i_{\kappa+2}}$};
        \draw[] (-2.5,-2.5) node [fill=white,rotate=-45] (x3_left) {$\vdots$};
        \draw[] (-3,-3) node [fill=white,text=blue] (xk) {$v_{i_k}$};
        \draw[] (-4.5,-4.5) node [fill=white] (left0) {$0$};
        
        % dots off main path
        \draw[] (1,0) node [fill=white,rotate=45] (x2_right_dots) {$\vdots$};
        \draw[] (-0.5,-1.5) node [fill=white,rotate=45] (x3_right_dots) {$\vdots$};
        \draw[] (3,-2) node [fill=white,rotate=45] (u2_dots) {$\vdots$};
        \draw[] (-2.5,-3.5) node [fill=white,rotate=45] (xk_dots) {$\vdots$};

        % second main path
        \draw[] (4,0) node [fill=white,text=teal] (x1_right) {$v_{3}$};
        \draw[] (2.5,-1.5) node [fill=white,text=teal] (x2_right) {$v_{11}$};
        \draw[] (1,-3) node [fill=white,rotate=-45] (x3_right) {$\vdots$};
        \draw[] (0.5,-3.5) node [fill=white,text=teal] (xk_right) {$v_{25}$};
        \draw[] (1,-4) node [fill=white,rotate=45] (xk_right_dots) {$\vdots$};
        \draw[] (-1,-5) node [fill=white] (xk_right0) {$0$};

        % third main path
        \draw[] (6,-2) node [fill=white,text=orange] (x1_right_right) {$v_{1}$};
        \draw[] (4.5,-3.5) node [fill=white,text=orange] (x2_right_right) {$v_2$};
        \draw[] (3,-5) node [fill=white,rotate=-45] (x3_right_right) {$\vdots$};
        \draw[] (2.5,-5.5) node [fill=white,text=orange] (xk_right_right) {$v_n$};
        \draw[] (1,-7) node [fill=white] (xk_right_right0) {$1$};

        % fourth main path (leafs)
        \draw[] (7,-3) node [fill=white] (x1_leaf) {$0$};
        \draw[] (5.5,-4.5) node [fill=white] (x2_leaf) {$0$};
        \draw[] (3.5,-6.5) node [fill=white] (xk_leaf) {$0$};

        % left branch arrows
        \draw[-Stealth] (x1) to (x2);
        \draw[-Stealth] (x2) to (x3);
        \draw[-Stealth] (x3) to (x3_left);
        \draw[-Stealth] (xk) to (left0);

        \draw[-Stealth] (x1_right) to (x2_right);
        \draw[-Stealth] (x2_right) to (x3_right);
        \draw[-Stealth] (xk_right) to (xk_right0);

        \draw[-Stealth] (x1_right_right) to (x2_right_right);
        \draw[-Stealth] (x2_right_right) to (x3_right_right);
        \draw[-Stealth] (xk_right_right) to (xk_right_right0);

        % right branch arrows
        \draw[-Stealth] (x1) to node[midway,above,sloped] {\footnotesize {Path $\pi\vert_{\oplus \kappa}$}} (x1_right);
        \draw[-Stealth] (x1_right) to (x1_right_right);

        \draw[-Stealth] (x1_right_right) to (x1_leaf);
        \draw[-Stealth] (x2_right_right) to (x2_leaf);
        \draw[-Stealth] (xk_right_right) to (xk_leaf);
    
        % labels of branches
        % \coordinate (labels) at (0,1);
        \coordinate (labels) at (-0.5,0.6);
        \draw[] (labels) node [fill=white,rotate=45,text=blue] (vc_label) {\footnotesize {Path $\pi$ corresponding to a vertex cover of $G$}};
        \draw[] ($(labels)+(2,-2)$) node [fill=white,rotate=45,text=teal] (neighbor_label) {\footnotesize {Vertices in $V_\kappa$}};
        \draw[] ($(labels)+(4,-4)$) node [fill=white,rotate=45,text=orange] (rest_label) {\footnotesize {Vertices in $V\setminus\{v_{i_\kappa},v_3\}$}};
        
    \end{tikzpicture}
    \caption{An illustration of the proof of the Yes case of~\Cref{claim:isedge-intro}. Given a vertex cover $C=\{v_{i_1},\ldots,v_{i_k}\}$ of $G$, our decision tree for $\mathrm{\isedge}$ queries $C$ among the leftmost branch (colored  blue in the figure). If some vertex $v_{i_\kappa}\in C$ is set to $1$ then the decision tree queries all vertices in $V_\kappa$ (colored teal). Once some $v\in V_\kappa$ is set to $1$, the decision tree queries the remaining unqueried vertices to ensure that they are set to $0$ (colored orange) before outputting $1$.}
    \label{fig:isedge upper bound}
\end{figure}
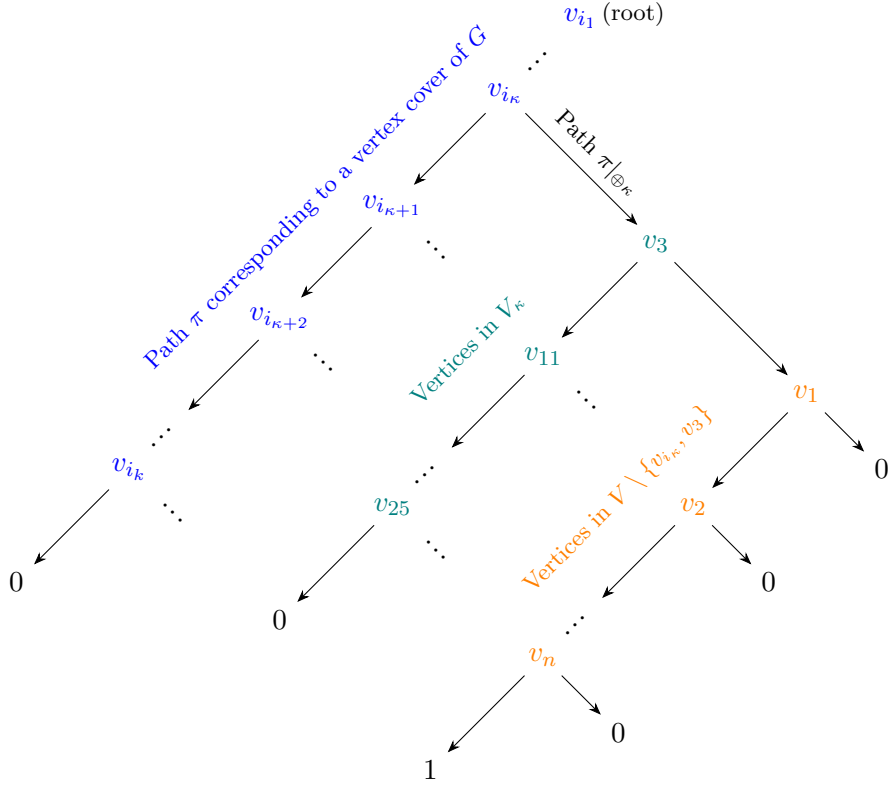

%% file: figisedgelowerbound.tex
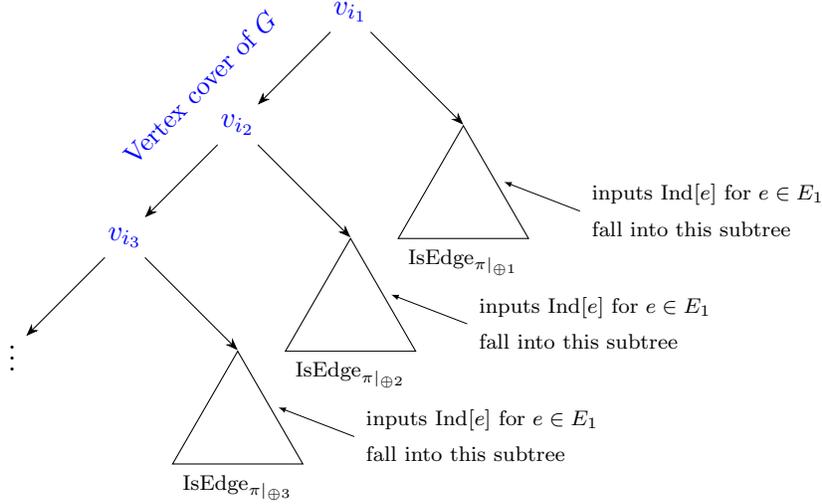
\begin{figure}[H]
    \centering
    \begin{tikzpicture}[]
        % points of left path
        \draw[] (2,2) node [fill=white,text=blue] (x1) {$v_{i_1}$};
        \draw[] (0.5,0.5) node [fill=white,text=blue] (x2) {$v_{i_2}$};
        \draw[] (-1,-1) node [fill=white,text=blue] (x3) {$v_{i_3}$};
        \draw[] (-2.5,-2.5) node [fill=white] (x4) {$\vdots$};

        %vdots label at the end of pi path
        % \draw[color=black] (x4) node [below] {$\vdots$};

        % arrows of path \pi
        \draw[-Stealth] (x1) to (x2);
        \draw[-Stealth] (x2) to (x3);
        \draw[-Stealth] (x3) to (x4);
        
        % subtrees
        \node[isosceles triangle,
            draw,
            anchor=apex,
            isosceles triangle apex angle=60,
            rotate=90,
            minimum size=1.5cm] (T1) at (3.5,0.5){};
        \node[isosceles triangle,
            draw,
            anchor=apex,
            isosceles triangle apex angle=60,
            rotate=90,
            minimum size=1.5cm] (T2) at (2,-1){};
        \node[isosceles triangle,
            draw,
            anchor=apex,
            isosceles triangle apex angle=60,
            rotate=90,
            minimum size=1.5cm] (T3) at (0.5,-2.5){};

        % subtree labels
        \draw[color=black] (T1.west) node [fill=white,below] {\scriptsize $\mathrm{\isedge}_{\pi\vert_{\oplus 1}}$};
        \draw[color=black] (T2.west) node [fill=white,below] {\scriptsize $\mathrm{\isedge}_{\pi\vert_{\oplus 2}}$};
        \draw[color=black] (T3.west) node [fill=white,below] {\scriptsize $\mathrm{\isedge}_{\pi\vert_{\oplus 3}}$};

%         \draw[black,-stealth] (k'M) to node[midway,above,sloped] {\footnotesize gap amplification} (k'l);
        % arrows to subtrees
        \draw[-Stealth] (x1) to (T1.apex);
        \draw[-Stealth] (x2) to (T2.apex);
        \draw[-Stealth] (x3) to (T3.apex);
        
        % forms a vertex cover label
        % \draw [black,decorate,decoration={brace,mirror,raise=10pt,amplitude=8pt}] (2.2,2.2) -- (x4) node [black,pos=0.5,yshift=0.8cm,above,sloped] {\small vertex cover of $G$};
        \draw[] (0,1) node [fill=white,rotate=45,text=blue] (vc_label) {\small {Vertex cover of $G$}};

        % input labels
        \coordinate (L) at (6.75,-1.25);
        \node[label={[align=left]{\scriptsize inputs $\ind[e]$ for $e\in E_1$}\\ {\scriptsize fall into this subtree}}] (e1) at (L){};
        \node[label={[align=left]{\scriptsize inputs $\ind[e]$ for $e\in E_1$}\\ {\scriptsize fall into this subtree}}] (e2) at ($(L)+(-1.5,-1.5)$){};
        \node[label={[align=left]{\scriptsize inputs $\ind[e]$ for $e\in E_1$}\\ {\scriptsize fall into this subtree}}] (e3) at ($(L)+(-3,-3)$){};

        % arrows to labels
        \coordinate (A) at (5.05,-.65);
        \draw[latex-,line width=0.01mm] ([xshift=0.1cm]T1.right side) to (A);
        \draw[latex-,line width=0.01mm] ([xshift=0.1cm]T2.right side) to ($ (A) + (-1.5,-1.5) $);
        \draw[latex-,line width=0.01mm] ([xshift=0.1cm]T3.right side) to ($ (A) + (-3,-3) $);
        
    \end{tikzpicture}
  \captionsetup{width=.9\linewidth}
    \caption{An illustration of the No case of~\Cref{claim:isedge-intro}. Given any decision tree $T$ computing $\mathrm{\isedge}$, the leftmost branch $\pi$ must form a vertex cover of $G$.  Furthermore, for each vertex $v_{i_\kappa}$ queried along $\pi$, the subtree $T_\kappa$ branching off of $\pi$ at $v_{i_{\kappa}}$ must query all the vertices in $V_\kappa$.  The size of $T$ is therefore at least $k' + \sum_{\kappa}|V_\kappa| = k'+ m.$ } \label{fig:isedge lower bound}
\end{figure}

%% file: figellisedgeupperbound.tex
\begin{figure}[H]
  \captionsetup{width=.9\linewidth}

    \centering
    \begin{tikzpicture}[tips=proper]
        % points of left path
        \draw[] (3,3) node [fill=white,text=blue] (root) {\small $v_{i_1}^{(0)}$};
        \draw[] (3.7, 3.05) node [fill=white] (rootlabel) {\footnotesize (root)};
        \draw[] (2.5,2.5) node [fill=white,rotate=-45] (incoming) {$\vdots$};
        \draw[] (2,2) node [text=blue] (x1) {{\small$v_{i_\kappa}^{(0)}$}};
        \draw[] (0.5,0.5) node [fill=white,text=blue] (x2) {\small$v_{i_{\kappa+1}}^{(0)}$};
        \draw[] (-1,-1) node [text=blue] (x3) {\small $v_{i_{\kappa+2}}^{(0)}$};
        \draw[] (-2.5,-2.5) node [rotate=-45] (x3_left) {$\vdots$};
        \draw[] (-3,-3) node [text=blue] (xk) {\small$v_{i_k}^{(0)}$};
        \draw[] (-4.5,-4.5) node [fill=white] (left0) {$0$};
        
        % dots off main path
        \draw[] (1,0) node [rotate=45] (x2_right_dots) {$\vdots$};
        \draw[] (-0.5,-1.5) node [rotate=45] (x3_right_dots) {$\vdots$};
        % \draw[] (3,-2) node [fill=white,rotate=45] (u2_dots) {$\vdots$};
        \draw[] (-2.5,-3.5) node [fill=white,rotate=45] (xk_dots) {$\vdots$};

        % k dup variables
        \draw[] (3.5,.5) node [fill=white,text=teal] (v1k) {\small$v_{_\kappa}^{(1)}$};
        \draw[] (4.5,-.5) node [rotate=45] (v1k_dots) {$\vdots$};
        \draw[] (5,-1) node [text=teal] (vlk) {\small$v_{_\kappa}^{(\ell)}$};

        % zeroes from dups
        \draw[] (2,-1) node [fill=white] (v1k_zero) {\small $0$};
        \draw[] (3.5,-2.5) node [fill=white] (vlk_zero) {\small $0$};

        % Vk vertices
        \draw[] (6,-2) node [fill=white,text=violet] (v03) {\small$v_3^{(0)}$};
        \draw[] (4.5,-3.5) node [fill=white,text=violet] (v011) {\small$v_{11}^{(0)}$};
        \draw[] (3.5,-4.5) node [rotate=-45] (dupdots) {$\vdots$};
        \draw[] (3,-5) node [text=violet] (v0n) {\small $v_{25}^{(0)}$};
        \draw[] (1.5,-6.5) node [text=black] (v0n_zero) {$0$};

        % Vk dots
        \draw[] (4.8,-3.8) node [rotate=45] (v011_dots) {\small $\vdots$};
        \draw[] (3.3,-5.3) node [rotate=45] (v0n_dots) {\small $\vdots$};
        
        % Vk dup variables
        % \draw[] (7,-3) node[text=cyan] (v13) {\small $v_3^{(1)}$};
        % \draw[] (8,-4) node[rotate=45] (vkdupdots) {$\vdots$};
        % \draw[] (8.5,-4.5) node[text=cyan] (vl3) {\small $v_3^{(\ell)}$};

        % zeroes from Vk dups
        % \draw[] (5.5,-4.5) node[] (v13_zero) {\small $0$};
        % \draw[] (7,-6) node[] (v1l_zero) {\small $0$};

        % rest of the variables
        \coordinate (rest) at (7.5,-3.5);
        \draw[] (rest) node[text=red] (v10) {\small $v_1^{(0)}$};
        \draw[] ($(rest)+(-1.5,-1.5)$) node[text=red] (v20) {\small $v_2^{(0)}$};
        \draw[] ($(rest)+(-3,-3)$) node[text=red] (vn0) {\small $v_n^{(0)}$};

        % zeros of RotV
        \draw[] ($(rest)+(1.5,-1.5)$) node[] (v10_zero) { $0$};
        \draw[] ($(rest)+(0,-3)$) node[] (v20_zero) { $0$};
        \draw[] ($(rest)+(-1.5,-4.5)$) node[] (vn0_zero) { $0$};

        % Vk dup variables
        \draw[] ($(rest)+(-4.5,-4.5)$) node[text=cyan] (v13) {\small $v_3^{(1)}$};
        \draw[] ($(rest)+(-3.5,-5.5)$) node[text=black,rotate=45] (vkdupdots) {$\vdots$};
        \draw[] ($(rest)+(-3,-6)$) node[text=cyan] (vl3) {\small $v_3^{(\ell)}$};

        % Vk leafs
        \draw[] ($(rest)+(-6,-6)$) node[text=black] (v13_zero) {0};
        \draw[] ($(rest)+(-4.5,-7.5)$) node[text=black] (vl3_zero) {0};
        \draw[] ($(rest)+(-1.5,-7.5)$) node[text=black] (vl3_one) {1};

        % left branch arrows
        \draw[-Stealth] (x1) to (x2);
        \draw[-Stealth] (x2) to (x3);
        \draw[-Stealth] (x3) to (x3_left);
        \draw[-Stealth] (xk) to (left0);

        \draw[-Stealth] (v1k) to (v1k_zero);
        \draw[-Stealth] (vlk) to (vlk_zero);

        \draw[-Stealth] (v03) to (v011);
        \draw[-Stealth] (v011) to (dupdots);
        \draw[-Stealth] (v0n) to (v0n_zero);

        \draw[-Stealth] (v10) to (v20);
        \draw[-Stealth] (v20) to (vn0);
        \draw[-Stealth] (vn0) to (v13);
        \draw[-Stealth] (v13) to (v13_zero);
        \draw[-Stealth] (vl3) to (vl3_zero);

        % right branch arrows
        \draw[-Stealth] (x1) to node[midway,above,sloped] {\scriptsize {Path $\pi\vert_{\oplus \kappa}$}} (v1k);

        \draw[-Stealth] (v1k) to (v1k_dots);
        \draw[-Stealth] (vlk) to (v03);
        
        \draw[-Stealth] (v03) to (v10);
        \draw[-Stealth] (v10) to (v10_zero);
        \draw[-Stealth] (v20) to (v20_zero);
        \draw[-Stealth] (vn0) to (vn0_zero);

        \draw[-Stealth] (v13) to (vkdupdots);
        \draw[-Stealth] (vl3) to (vl3_one);
    
        % labels of branches
        % \coordinate (labels) at (0,1);
        \coordinate (labels) at (-0.5,0.6);
        \draw[] (labels) node [fill=white,rotate=45,text=blue] (vc_label) {\footnotesize {Path $\pi$ corresponding to a vertex cover of $G$}};
        \draw[] (4.75,0) node [rotate=-45,text=teal] (dup1_label) {\footnotesize $\textsc{Dup}(v_{i_\kappa}^{(0)})$};
        \draw[] (4,-3.2) node [rotate=45,text=violet] (2_label) {\footnotesize Vertices from $V_\kappa$};
        \draw[] (5.5,-4.7) node [rotate=45,text=red] (3_label) {\footnotesize Original variables};
        \draw[] (4.5,-8.5) node [rotate=-45,text=cyan] (4_label) {\footnotesize $\textsc{Dup}(v_{3}^{(0)})$};
        % \draw[] ($(labels)+(2,-2)$) node [fill=white,rotate=45,text=teal] (neighbor_label) {\footnotesize {Vertices in $V_\kappa$}};
        % \draw[] ($(labels)+(4,-4)$) node [fill=white,rotate=45,text=orange] (rest_label) {\footnotesize {Vertices in $V\setminus\{v_{i_\kappa},v_3\}$}};
        
    \end{tikzpicture}
    \caption{An illustration of the proof of the Yes case of~\Cref{thm:ell-isedge-intro}. Given a vertex cover $C=\{v_{i_1},\ldots,v_{i_k}\}$ of $G$, our decision tree for $\ell\text{-}\mathrm{\isedge}$ queries the original variables corresponding to $C$ among the leftmost branch (colored  blue in the figure). If some vertex $v_{i_\kappa}^{(0)}\in C$ is set to $1$ then the decision tree queries all vertices in $\textsc{Dup}(v_{i_\kappa}^{(0)})$ (colored in teal). If all of these are $1$, then it proceeds to compute the appropriate $\mathrm{\isedge}$ subfunction on the remaining vertices. }
    \label{fig:ell isedge upper bound}
\end{figure}
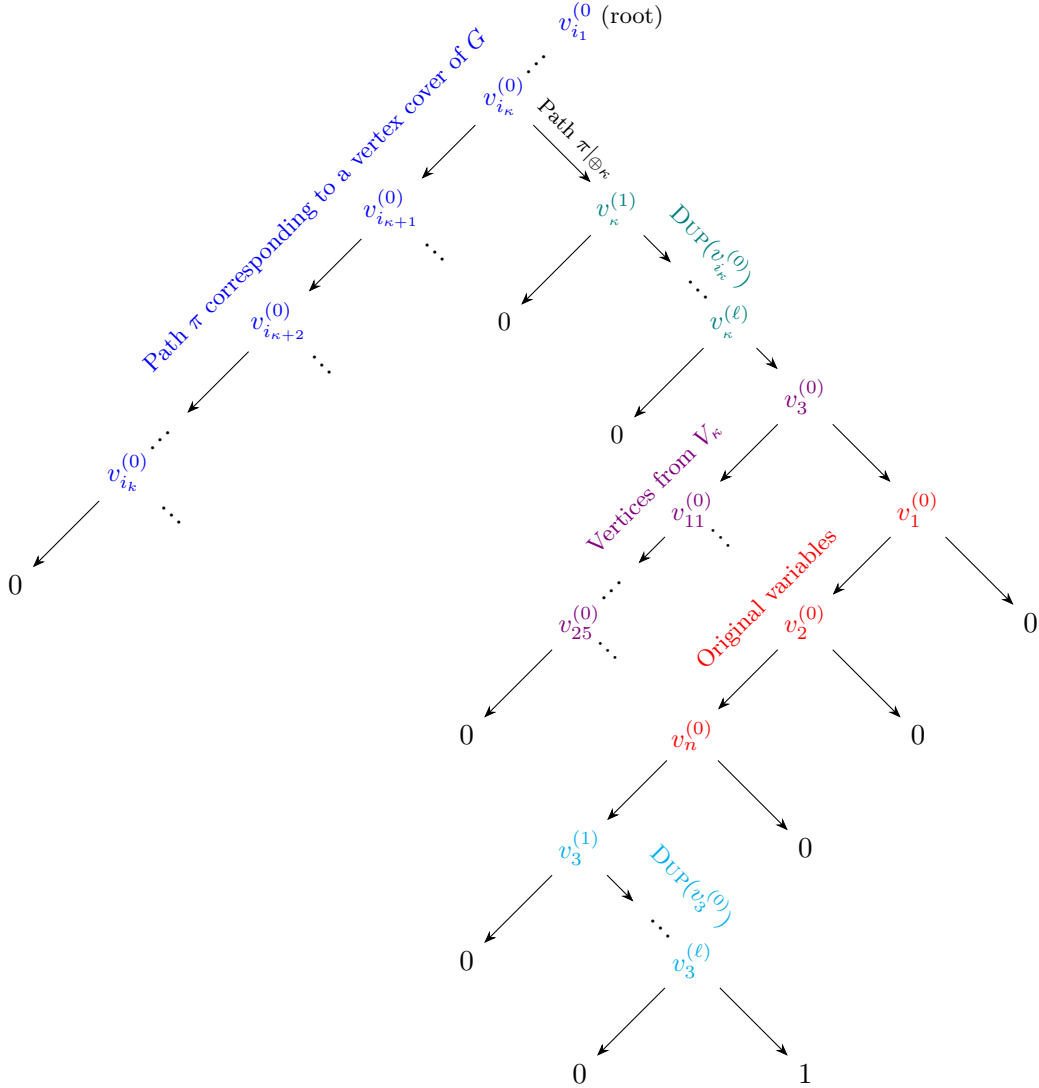

%% file: hardness_distillation.tex
\section{Hardness distillation and learning consequence for small error}

\subsection{A general method for hardness distillation}
\label{subsec:general hardness distillation}

For a function $f:\zo^n\to\zo$, the quantity $\dtsize(f)$ captures the complexity of computing $f$ on \textit{all} of its inputs. If $\dtsize(f)$ is large, then any small decision tree that tries to compute $f$ must err on at least one point in $\zo^n$. For some $f$, it may be the case that, more specifically, there is a fixed set $D\sse \zo^n$ such that all small decision trees err on some point in $D$. The set $D$ then captures or ``distills'' the hardness of $f$ since any function $g$ which agrees with $f$ over the set $D$ must also have large decision tree complexity. We call this set $D$ a \textit{coreset}.\footnote{This naming convention is inspired by, though not formally related to, the notion of a coreset from the clustering literature.} Ultimately, our goal will be to identify explicitly coresets $D$ which distill the hardness of the target function $f$. This way, any learner that learns $f$ over the distribution $\text{Uniform}(D)$ to error $< \frac{1}{|D|}$  has to output a decision tree whose size captures $\dtsize(f)$. Since the error scales with $\frac{1}{|D|}$, we have a vested interest in making $D$ has small as possible so that we can tolerate large learning errors. In this section, we identify a general method for distilling the hardness of a function $f$ into a coreset $D$. We start by generalizing certificate complexity and relevant variables with respect to fixed subsets $D$.

\paragraph{Certificate complexity with respect to a set of inputs.}{
A certificate for $f:\zo^n\to\zo$ over a set of inputs $D\sse \zo^n$ on $x\in \zo^n$ is a restriction $\rho$ consistent with $x$ such that $f_\rho(y)=f_\rho(x)$ for all $y\in D$. The certificate complexity of $x$ on $f$ over $D$ is the size of the smallest certificate of $x$ on $f$ over $D$.
}

One useful fact is that a decision tree path forms a certificates for the inputs that follow it. 
\begin{fact}[Decision tree paths are certificates]
\label{fact:DT path is a cert}
    If a decision tree $T$ computes $f:\zo^n\to\zo$ over $D\sse \zo^n$, then the path that an input $x\in D$ follows in $T$ forms a certificate for $f$ over $D$ on $x$. 
\end{fact}
Indeed, in the above, any root-to-leaf path $\pi$ terminates in a leaf which implies $f_{\pi}$ is a constant function over $D$. Any input $x\in D$ that follows $\pi$ is consistent with it and so $f(x)=f_\pi(x)=f_{\pi}(y)$ for all $y\in D$.

\paragraph{Relevant variables with respect to a set of inputs.}{
A variable $i\in [n]$ is said to be \textit{relevant} for $f:\zo^n\to\zo$ over a set of inputs $D\sse \zo^n$ if there is some $x\in D$ such that $x^{\oplus i}\in D$ and $f(x)\neq f(x^{\oplus i})$. We write $\Rel(f; D)\in [n]$ for the number of relevant variables of $f$ with respect to $D$. When referring to the number of relevant variables over the entire domain $D=\zo^n$, we drop $D$ and simply write $\Rel(f)$. If $D\sse D'$, then every relevant variable for $f$ over $D$ is also relevant for $f$ over $D'$. Therefore, $\Rel(f,D)\le \Rel(f,D')$ and in particular $\Rel(f,D)\le \Rel(f)$ for all $D$. 
}

\paragraph{Decision tree complexity with respect to a set of inputs.}{The decision tree complexity of $f:\zo^n\to\zo$ over $D\sse \zo^n$ is the size of the smallest decision tree that computes $f$ over $D$ and is denoted $\dtsize(f,D)$. Any decision tree that computes $f$ also computes $f$ over $D$ and so $\dtsize(f,D)\le\dtsize(f)$.}

% \paragraph{Hardness distillation.}\lnote{Let's have a more informative paragraph title -- this entire section is about hardness distillation.} {
% Relevant variables are useful for lower bounding decision tree size since any decision tree computing $f$ over $D$ must necessarily query all $\Rel(f;D)$ relevant variables. For hardness distillation, we would like to identify a subset $D$ such that $\dtsize(f)\approx \dtsize(f,D)$\lnote{Should we just write $=$ instead of $\approx$?}. To this end, we establish the following relationship between $\dtsize(f,D)$ and relevant variables of subfunctions of $f$\lnote{This prose should also mention certificate complexity}:
% }
We now state and prove the main result for this section.
\begin{theorem}[Hardness distillation]
\label{thm:hardness distillation}
    Let $f:\zo^n\to\zo$ and $D\sse \zo^n$ be a set of inputs. Let $s_1\in \N$ lower bound the certificate complexity of $x$ on $f$ over $D$. Let $s_2\in\N$ satisfy
    $$
    \sum_{i=1}^{|\rho|}\Rel(f_{\pi\vert_{\oplus i}};D)\ge s_2
    $$
    for every certificate $\rho$ for $x$ on $f$ over $D$ and $\pi\in \mathrm{Perm}(\rho)$, a permutation of $\rho$. Then,
    $$
    \dtsize(f,D)\ge s_1+s_2.
    $$
\end{theorem}
\input{fighardnessdistillation}
If one can show for some $D$ that the quantity $s_1+s_2$ captures the decision tree complexity of $f$, then $D$ is a good candidate for hardness distillation. \Cref{fig:hardness distillation} illustrates some intuition for the quantity $\sum_{i=1}^{|\rho|}\Rel(f_{\pi\vert_{\oplus i}};D)$ in \Cref{thm:hardness distillation}. If $x\in D$, then any decision tree for $f$ over $D$ contains a certificate, $\rho$, for $f$ on $x$. The depth-$|\rho|$ path followed by $x$ induces an ordering over $\rho$ and naturally yields $|\rho|$ disjoint subtrees, each of which hangs off the main path. The size of the main decision tree is lower bounded by the sizes of these subtrees plus the length of the path followed by $x$. The sizes of these subtrees can be lower bounded by the number of relevant variables of the corresponding subfunctions which then yields the desired lower bound.

Before proving \Cref{thm:hardness distillation}, we establish a lemma stating that the number of relevant variables of disjoint subtrees of a decision tree lower bounds its size. 

\begin{lemma}[Relevant variables of disjoint subtrees lower bound decision tree size]
\label{lemma:relevant variables lb dt size}
    Let $T$ be a decision tree, and let $T_1,\ldots,T_d$ be disjoint subtrees of $T$. Then,
    $$
    |T|\ge \sum_{i=1}^d\Rel(T_i).
    $$
\end{lemma}

\begin{proof}
    If a variable $x_j$ is \textit{not} queried in the subtree $T_i$, then $x_j$ cannot be relevant for the function $T_i$. Indeed, in this case, the leaf in $T_i$ that any input $x$ reaches is the same as the leaf that $x^{\oplus j}$ reaches. Therefore, every relevant variable of $T_i$ is queried in the subtree. Since the subtrees $T_1,\ldots,T_d$ are disjoint, each relevant variable of $T_i$ can be mapped to a \textit{unique} internal node of $T$. It follows that 
    \[ 
    |T|\ge\sum_{i=1}^d |T_i| \ge \sum_{i=1}^d\Rel(T_i).
    \qedhere \]  
\end{proof}
With this lemma in hand, we are able to prove \Cref{thm:hardness distillation}.
\begin{proof}[Proof of \Cref{thm:hardness distillation}]
    Let $T$ be any decision tree computing $f$ over $D$. We will show that $|T|\ge s_1+s_2$.  Let $\pi$ be the path followed by $x\in D$ in $T$. By \Cref{fact:DT path is a cert}, $\pi$ is a certificate for $f$ over $D$ on $x$. Therefore $|\pi|\ge s_1$. Recall from \Cref{defn:path prefix} that $\pi\vert_{\oplus i}=\left\{\pi(1),\ldots,\pi({i-1}),\overline{\pi(i)}\right\}$ corresponds to the depth $i$ path in $T$ that follows $x$ to depth $i-1$ and then diverges from $x$ on the $i$th variable queried. Let $T_{\pi\vert_{\oplus i}}$ denote the subfunction of $T$ computed by the subtree rooted at the last variable queried in $\pi\vert_{\oplus i}$. Each $T_{\pi\vert_{\oplus i}}$ contributes $\Rel(T_{\pi\vert_{\oplus i}})$ many variables to the size of $T$ by \Cref{lemma:relevant variables lb dt size} and the path $\rho$ itself contributes at least $s_1$ many variables since $\pi$ is also disjoint from the subtrees. It follows that
    \begin{align*}
        |T|&\ge |\pi|+\sum_{i=1}^{|\pi|} \Rel(T_{\pi\vert_{\oplus i}})\tag{\Cref{lemma:relevant variables lb dt size}}\\
        &\ge s_1+\sum_{i=1}^{|\pi|} \Rel(T_{\pi\vert_{\oplus i}}; D)\tag{Definition of $\Rel$}\\
        &=s_1+\sum_{i=1}^{|\pi|} \Rel(f_{\pi\vert_{\oplus i}}; D)\tag{$T$ computes $f$ over $D$}\\
        &\ge s_1+s_2\tag{Assumption from theorem statement}.
    \end{align*}
\end{proof}

\subsection{Warmup: hardness distillation for ${\mathrm{\sc IsEdge}}$}

\label{subsec:hardness distillation for isedge}

We start by applying the framework from \Cref{subsec:general hardness distillation} to the function \isedge. The first step is to identify a small coreset $D$ which captures decision tree size.

\begin{definition}[Decision tree coreset of the ${\mathrm{\sc IsEdge}}$ function]
For an $n$-vertex graph $G$, the set $D_G\sse \zo^n$ consists of the points
\begin{itemize}
    \item all edge indicators: $\ind[e]\in \zo^n$ such that $e\in E$;
    \item all 1-coordinate perturbations of edge indicators: $\ind[e]^{\oplus i}$ and $\ind[e]^{\oplus j}$ for all $e=\{v_i,v_{j}\}\in E$;
    \item the all 0s inputs: $0^n$.
\end{itemize}
\end{definition}

\paragraph{Example.}{See \Cref{fig:example of G and DG} for an example of a graph $G$ and the associated set of inputs $D_G$.
\input{figcoresetofisedge}
}

Recall from \Cref{claim:isedge-intro} that $\dtsize(\mathrm{\isedge}_G)\ge k'+m$ where $k'$ is the size of a vertex cover for $G$. The main claim of this section is that $D_G$ ``distills'' this hardness factor of $\mathrm{\isedge}_G$. The upper bound from \Cref{claim:isedge-intro} immediately applies to $D_G$. That is, $k+m+mn\ge \dtsize(\mathrm{\isedge}_G)\ge \dtsize(\mathrm{\isedge}_G, D_G)$. Therefore, the lower bound is all that remains for showing $D_G$ is a good coreset.  

\begin{claim}[$D_G$ is a decision tree coreset for ${\mathrm{\sc IsEdge}}$]
\label{claim:D_G is a certificate of large DT size}
Let $G$ be an $m$-vertex graph. Then
$$
\dtsize(\mathrm{\isedge},D_G)\ge 
k'+m
$$
where $k'$ is the size of a vertex cover for $G$.
\end{claim}

Ultimately, we would like to prove \Cref{claim:D_G is a certificate of large DT size} by applying \Cref{thm:hardness distillation} where $f$ is the function $\mathrm{\isedge}$, $D$ is the set of inputs $D_G$, and $x$ is the input $0^n$. To this end, we first show that $k'$ is a lower bound on the certificate complexity of $0^n$ on $\mathrm{\isedge}$ over $D_G$. Then we prove a lemma showing that the number of edges in $G$ lower bounds the number of relevant variables of subfunctions of $\mathrm{\isedge}$ induced by certificates of $0^n$. 

\begin{proposition}[Any certificate of $0^n$ contains a vertex cover]
\label{prop:cert_of_0n_is_a_vc}
    Let $G$ be an $n$-vertex graph and let $\rho$ be a certificate for $\mathrm{\isedge}$ over $D_G$ on $0^n$. Then the variables in $\rho$ form a vertex cover of $G$.
\end{proposition}
\begin{proof}
    If the variables in $\rho$ do not cover some edge $e\in E$, then $\ind[e]\in \zo^n$ is consistent with $\rho$ and $\mathrm{\isedge}_G(0^n)=0\neq 1=\mathrm{\isedge}_G(\ind[e])$ implies $\rho$ is not a certificate. Therefore, any certificate $\rho$ must contain a vertex cover.    
\end{proof}

\begin{lemma}[Lower bounding the number of relevant variables of $\mathrm{\isedge}$ subfunctions]
\label{lemma:covered edges lemma}
    Let $G$ be an $n$-vertex graph, $\rho$ a certificate for $\mathrm{\isedge}:\zo^n\to\zo$ over $D_G$, and $\pi=(\overline{v}_{i_1},\ldots,\overline{v}_{i_k})\in\mathrm{Perm}(\rho)$ a permutation of $\rho$. Then
    $$
    \Rel(\mathrm{\isedge}_{\pi\vert_{\oplus \kappa}}; D_G)\ge |E({v_{i_\kappa}; v_{i_1},\ldots,v_{i_{\kappa-1}}})|
    $$
    for all $\kappa\in [k]$.
\end{lemma}

\begin{proof}
    Let $\pi\vert_{\oplus \kappa}$ be as in the lemma statement and let $v\in V_\kappa=V(v_{i_\kappa};v_{i_1},\ldots,v_{i_{\kappa-1}})$ be arbitrary (recall the definition of these quantities from \Cref{defn:path prefix,defn:vertex neighborhood}). Let $e=(v_{i_\kappa},v)\in E_\kappa=E({v_{i_\kappa}; v_{i_1},\ldots,v_{i_{\kappa-1}}})$ be the edge containing $v$. The input $\ind[e]\in D_G$ has a $1$ for the coordinates corresponding to $v$ and $v_{i_\kappa}$ and $0$s elsewhere. Therefore, it is consistent with $\pi\vert_{\oplus \kappa}=\{\overline{v}_{i_1},\ldots,\overline{v}_{i_{\kappa-1}},v_{i_\kappa}\}$ (since $v\not\in\{v_{i_1},\ldots,v_{i_{\kappa-1}}\}$ by the definition of $V_\kappa$). The input $\ind[e]^{\oplus v}\in D_G$ is similarly consistent with $\pi\vert_{\oplus\kappa}$. Therefore, each $v\in V_\kappa$ is a distinct relevant variable for $\mathrm{\isedge}_{\pi\vert_{\oplus \kappa}}$ over $D_G$:
     $$
    \mathrm{\isedge}_{\pi\vert_{\oplus \kappa}}\left(\ind[e]\right)=1\quad\text{and}\quad \mathrm{\isedge}_{\pi\vert_{\oplus \kappa}}\left(\ind[e]^{\oplus v}\right)=0.
    $$   
    It follows that $\Rel(\mathrm{\isedge}_{\pi\vert_{\oplus \kappa}}; D_G)\ge |V_\kappa|=|E_\kappa|$ as desired.
\end{proof}

\begin{proof}[Proof of \Cref{claim:D_G is a certificate of large DT size}]
Let $\rho$ be a certificate for $\mathrm{\isedge}$ over $D_G$ on $0^n$. By \Cref{prop:cert_of_0n_is_a_vc}, the variables of $\rho$ form a vertex cover and so $|\rho|\ge k'$ where $k'$ is the size of a vertex cover of $G$. Let $\pi=(\overline{v}_{i_1},\ldots,\overline{v}_{i_k'})\in\mathrm{Perm}(\rho)$ be an arbitrary permutation of $\rho$. Then:
\begin{align*}
    \sum_{\kappa=1}^{|\pi|}\Rel(\mathrm{\isedge}_{\pi\vert_{\oplus j}};D_G)&\ge \sum_{\kappa=1}^{|\pi|}|E({v_{i_\kappa}; v_{i_1},\ldots,v_{i_{\kappa-1}}})|\tag{\Cref{lemma:covered edges lemma}}\\
    &= m\tag{\Cref{fact:restricted vertex neighborhood partitions}}.
\end{align*}
It follows from \Cref{thm:hardness distillation} that $\dtsize(\mathrm{\isedge},D_G)\ge k'+m$. 
\end{proof}

\subsection{Hardness distillation for $\ell$-$\mathrm{\isedge}$}

Following the ideas from \Cref{subsec:hardness distillation for isedge}, we show that the following set of inputs forms a coreset of $\ell$-{\isedge}. 

\begin{definition}[Coreset for $\ell$-{\isedge}]
    For an $n$-vertex, $m$-edge graph $G$ and $\ell\in\N$, the set $\ell\text{-}D_{G}\sse \zo^{n}\times (\zo^\ell)^n$ consists of the $m+m(2\ell+2)+1$ many points
    \begin{itemize}
        \item all generalized edge indicators:  $\ell\text{-}\ind[e]\in (\zo^n)^{\ell+1}$ for each edge $e\in E$ where $\ell\text{-}\ind[e]\coloneqq (\ind[e])^{\ell+1}$;
        \item $1$-coordinate perturbations of edge indicators: $2\ell+2$ many points for each $e\in E$ obtained by flipping one of the $1$-coordinates in $\ell\text{-}\ind[e]$; and
        \item the all $0$s input: $0^{n\ell+n}$.
    \end{itemize}
\end{definition}

\paragraph{Example.}{See \Cref{fig:example of G and ell-DG} for an example of a graph $G$ and the associated set of inputs $\ell\text{-}D_G$.
\input{figellisedgecoreset}
}

Recall from \Cref{thm:ell-isedge-intro} that $\dtsize(\ell\text{-}{\mathrm{\sc IsEdge}})\ge (\ell+1)\cdot (k'+m)$ where $k'$ is the size of a vertex cover for $G$. The main claim of this section is that $\ell\text{-}D_G$ distills this hardness factor of $\ell\text{-}{\mathrm{\sc IsEdge}}$. 

\begin{claim}[$\ell\text{-}D_{G}$ is a coreset for $\ell\text{-}{\mathrm{\sc IsEdge}}$]
\label{claim:ell D_G is a certificate of large DT size}
Let $G$ be an $m$-vertex graph and $\ell\in\N$ be arbitrary. Then,
$$
\dtsize(\ell\text{-}\mathrm{\isedge}, \ell\text{-}D_{G})\ge (\ell+1)(k'+m)
$$
where $k'$ is the size of a vertex cover for $G$.
\end{claim}

This claim is analgous to \Cref{claim:D_G is a certificate of large DT size} and the proof similarly proceeds in two steps. Ultimately, we will apply \Cref{thm:hardness distillation} where $f$ is $\ell\text{-}\mathrm{\isedge}:\zo^N\to\zo$, $D$ is $\ell\text{-}D_G$, and $x$ is $0^{N}$. As such, the first step extends \Cref{prop:cert_of_0n_is_a_vc} to $\ell\text{-}\mathrm{\isedge}$ and shows that certificates for $0^N$ contain vertex covers. The second step extends \Cref{lemma:covered edges lemma} and lower bounds the number of relevant variables of subfunctions of $\ell\text{-}\mathrm{\isedge}$ induced by certificates of $0^N$.

\begin{proposition}[Any certificate of $0^N$ contains a vertex cover]
\label{prop:certificate of 0N is a vc}
    Let $G$ be a graph and let $\{\overline{v}_{i_1}^{(j_1)},\ldots,\overline{v}_{i_k}^{(j_k)}\}$ be a certificate for $\ell\text{-}\mathrm{\isedge}$ over $\ell\text{-}D_G$ on $0^N$. Then, the vertices $\{v_{i_1},\ldots,v_{i_k}\}$ form a vertex cover of $G$. 
\end{proposition}
\begin{proof}
    If an edge $e$ is not covered by the vertices $\{v_{i_1},\ldots,v_{i_k}\}$, then the $1$-input $\ell\text{-}\ind[e]$ is consistent with any restriction of the form $\rho=\{\overline{v}_{i_1}^{(j_1)},\ldots,\overline{v}_{i_k}^{(j_k)}\}$. Therefore, any such $\rho$ cannot be a certificate.  
\end{proof}

\begin{lemma}[Lower bounding the number of relevant variables of $\ell\text{-}\mathrm{\isedge}$ subfunctions]
\label{lemma:covered edges lemma extended for ell-isedge}
Let $\rho$ be a certificate for $\ell\text{-}\mathrm{\isedge}$ over $\ell\text{-}D_G$ on $0^N$ and $\pi=(\overline{v}_{i_1}^{(j_1)},\ldots,\overline{v}_{i_k}^{(j_k)})\in\mathrm{Perm}(\rho)$, a permutation of $\rho$. Then
$$
\Rel(\ell\text{-}\mathrm{\isedge}_{\pi\vert_{\oplus\kappa}},\ell\text{-}D_G)\ge \ell+(\ell+1)\cdot |E(v_{i_\kappa};v_{i_1},\ldots,v_{i_{\kappa-1}})|
$$
for all $\kappa\in [k]$ such that $E(v_{i_\kappa};v_{i_1},\ldots,v_{i_{\kappa-1}})\neq \varnothing$. 
\end{lemma}

\begin{proof}
Let $\pi\vert_{\oplus\kappa}$ be as in the lemma statement and let $E_\kappa$ and $V_\kappa$ denote $E(v_{i_\kappa};v_{i_1},\ldots,v_{i_{\kappa-1}})$ and $V(v_{i_\kappa};v_{i_1},\ldots,v_{i_{\kappa-1}})$, respectively (recall these quantities from \Cref{defn:path prefix,defn:vertex neighborhood}). If $E_\kappa\neq\varnothing$, then we will show that $v_{i_\kappa}$ contributes $\ell$ relevant variables to $\Rel(\ell\text{-}\mathrm{\isedge}_{\pi\vert_{\oplus\kappa}},\ell\text{-}D_G)$ and that each $v\in V_\kappa$ contributes $\ell+1$.

\paragraph{The vertex $v_{i_\kappa}$ contributes $\ell$ relevant variables.}{
By assumption, $E_\kappa$ is nonempty so there is some edge $e=(v_{i_\kappa},v)\in E_\kappa$. The restriction $\pi\vert_{\oplus\kappa}$ sets one coordinate, $v_{i_\kappa}^{j_\kappa}$, to $1$ and the other coordinates: $\{v_{i_1},\ldots,v_{i_{\kappa-1}}\}$ are set to $0$. Since $v\not\in \{v_{i_1},\ldots,v_{i_{\kappa-1}}\}$, the input $\ell\text{-}\ind[e]\in\zo^N$ is consistent with $\pi\vert_{\oplus\kappa}$. All of the coordinates in $\textsc{Dup}(v_{i_\kappa})\cup\{v_{i_\kappa}^{(0)}\}$ are set to $1$ in the input $\ell\text{-}\ind[e]$. Hence, for any $v'\in \textsc{Dup}(v_{i_\kappa})\cup\{v_{i_k}^{(0)}\}\setminus\{v_{i_\kappa}^{(j_\kappa)}\}$, the input $\ell\text{-}\ind[e]^{\oplus v'}$ is consistent with $\pi\vert_{\oplus\kappa}$ since $v'\neq v_{i_\kappa}^{(j_\kappa)}$. Therefore,
$$
    \ell\text{-}\mathrm{\isedge}_{\pi\vert_{\oplus\kappa}}(\ell\text{-}\ind[e])=1\quad \text{and}\quad \ell\text{-}\mathrm{\isedge}_{\pi\vert_{\oplus\kappa}}(\ell\text{-}\ind[e]^{\oplus v'})=0
$$
and $\ell\text{-}\ind[e],\ell\text{-}\ind[e]^{\oplus v'}\in\ell\text{-}D_G$. Since $v'$ was arbitrary this shows that each of the $\ell$ variables in \textsc{Dup}$(v_{i_\kappa})\cup\{v_{i_\kappa}^{(0)}\}\setminus\{v_{i_\kappa}^{(j_\kappa)}\}$ is relevant for $\ell\text{-}\mathrm{\isedge}_{\pi\vert_{\oplus\kappa}}$ over $\ell\text{-}D_G$. 
}

\paragraph{Each vertex $v\in V_\kappa$ contributes $\ell+1$ relevant variables.}{
Let $v\in V_\kappa$ be an arbitrary vertex and let $e=(v_{i_\kappa},v)\in E_\kappa$ be the edge incident to $v_{i_\kappa}$ that contains $v$. Let $v'\in \textsc{Dup}(v)\cup\{v^{(0)}\}$ be a coordinate of $\ell\text{-}\mathrm{\isedge}$. As above, the inputs $\ell\text{-}\ind[e]$ and $\ell\text{-}\ind[e]^{\oplus v'}$ are both consistent with the restriction $\pi\vert_{\oplus\kappa}$. Moreover,
$$
    \ell\text{-}\mathrm{\isedge}_{\pi\vert_{\oplus\kappa}}(\ell\text{-}\ind[e])=1\quad \text{and}\quad \ell\text{-}\mathrm{\isedge}_{\pi\vert_{\oplus\kappa}}(\ell\text{-}\ind[e]^{\oplus v'})=0
$$
and $\ell\text{-}\ind[e],\ell\text{-}\ind[e]^{\oplus v'}\in\ell\text{-}D_G$. This shows that all $\ell+1$ variables in $\textsc{Dup}(v)\cup\{v^{(0)}\}$ for $v\in V_\kappa$ is relevant.
}
All of these relevant variables are unique and so the total number of relevant variables of $\ell\text{-}\mathrm{\isedge}_{\pi\vert_{\oplus\kappa}}$ is at least $\ell+(\ell+1)|V_\kappa|=\ell+(\ell+1)|E_\kappa|$ as desired. 
\end{proof}

\begin{proof}[Proof of \Cref{claim:ell D_G is a certificate of large DT size}]
Let $\rho$ be a certificate for $\ell\text{-}\mathrm{\isedge}$ over $\ell\text{-}D_G$ on $0^N$. By \Cref{prop:certificate of 0N is a vc}, the variables in $\rho$ form a vertex cover and so $|\rho|\ge$ the size of a vertex cover of $G$. Let $\pi=(\overline{v}_{i_1}^{(j_1)},\ldots,\overline{v}_{i_k}^{(j_k)})\in\mathrm{Perm}(\rho)$ be a permutation of $\rho$. In order to apply hardness distillation (\Cref{thm:hardness distillation}), we need to lower bound the number of relevant variables of $\ell\text{-}\mathrm{\isedge}_{\pi\vert_{\oplus \kappa}}$ for $\kappa=1,2,\ldots,|\pi|$. However, the lower bound from \Cref{lemma:covered edges lemma extended for ell-isedge} only applies if the corresponding restricted edge neighborhood $E_\kappa=E(v_{i_\kappa};v_{i_1},\ldots,v_{i_{\kappa-1}})$ is nonempty. To this end, we consider the restriction $\rho'=\{\overline{v}_{i_\kappa}^{(j_\kappa)}\mid E_\kappa\neq\varnothing\}\sse \rho$. This restriction is still a certificate for $\ell\text{-}\mathrm{\isedge}$ over $\ell\text{-}D_G$ on $0^N$ and therefore must still contain a vertex cover by \Cref{prop:certificate of 0N is a vc}. Therefore, $|\rho'|\ge k'$ where $k'$ is the size of a vertex cover of $G$. We can now write
\begin{align*}
    \sum_{\kappa=1}^{|\pi|}\Rel(\ell\text{-}\mathrm{\isedge};\ell\text{-}D_G)&\ge \sum_{\substack{\kappa\in [|\pi|]\\ E_\kappa\neq \varnothing}}\Rel(\ell\text{-}\mathrm{\isedge}_{\pi\mid_{\oplus \kappa}};\ell\text{-}D_G)\\
    &\ge \sum_{\substack{\kappa\in [|\pi|]\\ E_\kappa\neq \varnothing}} \ell+(\ell+1) |E_\kappa|\tag{\Cref{lemma:covered edges lemma extended for ell-isedge}}\\
    &=|\rho'|\ell+(\ell+1) m\tag{\Cref{fact:restricted vertex neighborhood partitions}: $\{E_\kappa\}$ partition $E$}\\
    &\ge \ell k'+(\ell+1) m\tag{$\rho'$ contains a vertex cover}.
\end{align*}
We have satisfied the conditions of \Cref{thm:hardness distillation} with $f$ being $\ell\text{-}\mathrm{\isedge}$, $D$ being $\ell\text{-}D_G$, and $x$ being $0^N$. We conclude
\[
\dtsize(\ell\text{-}\mathrm{\isedge}, \ell\text{-}D_G)\ge k'+k'\ell+(\ell+1) m= (\ell+1) (k'+m).
\qedhere
\]
\end{proof}

\subsection{Learning consequence for inverse polynomial error} 

In this section, we use \Cref{claim:ell D_G is a certificate of large DT size} to obtain hardness of learning decision trees with membership queries. We recall, formally, the learning problem we are interested in.\medskip 

\begin{tcolorbox}[colback = white,arc=1mm, boxrule=0.25mm]
\vspace{3pt} 
\textsc{DT-Learn}$(n,s,s',\eps)$:  Given random examples from an unknown distribution $\mathcal{D}$ and membership queries to a size-$s$ target decision tree, output a size-$s'$ decision tree which $\eps$-approximates the target over $\mathcal{D}$.
\vspace{3pt} 
\end{tcolorbox}
\medskip                                                                                                                                               
\begin{theorem}[Hardness learning DTs with inverse polynomial error]
\label{thm:learning hardness inverse poly error}
For all constants $\delta'>0, d\in \N$, there is a sufficiently small constant $\delta>0$ such that the following holds. If \textsc{DT-Learn}$(n,s,(1+\delta)\cdot s,\eps)$ with $s=O(n)$, $\eps=O(1/n)$ can be solved in randomized time $t(n)$, then $\textsc{VertexCover}(k,(1+\delta')\cdot k)$ on degree-$d$, $n$-vertex graphs can be solved in randomized time $O(n^2\cdot t(n^2))$.
\end{theorem}

\begin{proof}
    Given $\delta'>1$ and $d\in\N$, let $\lambda<1$ be any large enough constant so that $\lambda (1+\delta')>1$ and let $\delta>0$ be any constant satisfying $1<(1+\delta)<\min\{\lambda (1+\delta'), 1+\frac{1-\lambda}{d}\}$. Then we will use a learner for \textsc{DT-Learn}$(n,s,(1+\delta)\cdot s,\eps)$ to solve \textsc{VertexCover}$(k,(1+\delta') k)$.

 \paragraph{The reduction.}{   
    Fix $\ell=\Theta(n)$ large enough so that $1+\frac{1-\lambda}{d}>(1+\delta)+\frac{2(1+\delta)n}{\ell}$. Such an $\ell$ exists since $1+\frac{1-\lambda}{d}>1+\delta$ by assumption. Consider the function $\ell\text{-}{\mathrm{\isedge}}:\zo^N\to\zo$ and the set of inputs $\ell\text{-}D_G$ for $N=n+\ell n=\Theta(n^2)$. Let $\mathcal{D}$ be the distribution which is uniform over the set $\ell\text{-}D_G$ and fix $\eps<1/|\supp(\ell\text{-}D_G)|=O(1/m^2)$ (which is $O(1/N)$ since $n=\Theta(m)$ for constant degree graphs) and $s=\ell(k+m)+2mn=O(N)$. Run the procedure in \Cref{fig:solving vc with dt learn}. 
\begin{figure}[h!]
\begin{tcolorbox}[colback = white,arc=1mm, boxrule=0.25mm]
\vspace{3pt}
$\textsc{VertexCover}(k,(1+\delta')\cdot k)$:
\begin{itemize}[leftmargin=10pt,align=left]
\item[\textbf{Given:}] $G$, an $m$-edge degree-$d$ graph over $n$ vertices and $k\in \N$
\item[\textbf{Run:}]\textsc{DT-Learn}$(N,s,(1+\delta)\cdot s,\eps)$ for $t(N)$ time steps providing the learner with
\begin{itemize}[align=left,labelsep*=0pt]
    \item \textit{queries}: return $\ell\text{-}{\mathrm{\isedge}}(v^{(0)},\ldots,v^{(\ell)})$ for a query $(v^{(0)},\ldots,v^{(\ell)})\in\zo^N$; and
    \item \textit{random samples}: return $(\bv^{(0)},\ldots,\bv^{(\ell)})\sim\mathcal{D}$ for a random sample.
\end{itemize}
\item[$T_{\text{hyp}}\leftarrow$ decision tree output of the learner]
\item[$\eps_{\text{hyp}}\leftarrow \dist_{\mathcal{D}}(T_{\text{hyp}},\ell\text{-}{\mathrm{\isedge}})$]
\item[\textbf{Output:}] \textsc{Yes} if and only if $|T_{\text{hyp}}|\le (1+\delta)\cdot \left[\ell(k+m)+2mn\right]$ and $\eps_{\text{hyp}}\le \eps$
\end{itemize}
\vspace{3pt}
\end{tcolorbox}
\medskip
\caption{Using an algorithm for \textsc{DT-Learn} to solve \textsc{VertexCover}.}
\label{fig:solving vc with dt learn}
\end{figure}
}

\paragraph{Runtime.}{ 
    Any query $(v^{(0)},\ldots,v^{(\ell)})\in\zo^N$ to $\ell\text{-}{\mathrm{\isedge}}$ can be answered in $O(N)$ time by looking at $G$ and computing $\mathrm{\isedge}(v^{(0)})$ in time $O(m)$ then checking that the appropriate vertices are set to $1$. Similarly, a random sample from $\mathcal{D}$ can be obtained in time $O(N)$ by picking a uniform random element of $\ell\text{-}D_G$. This algorithm for \textsc{VertexCover} requires $O(N\cdot t(N))$ time to run the learner plus time $O(N^2)$ to compute $\dist_{\mathcal{D}}(T_{\text{hyp}},\ell\text{-}{\mathrm{\isedge}})$. Since $t(N)\ge N$, this implies an overall runtime of $O(N\cdot t(N))$ which is $O(n^2\cdot t(n^2))$.
}

\paragraph{Correctness.}{For correctness, we analyze the yes and no cases separately.
\subparagraph{Yes case: $\VC(G)\le k$.}{
    In this case, \Cref{thm:ell-isedge-intro} ensures that
    $$
    \dtsize(\ell\text{-}{\mathrm{\isedge}})\le (\ell+1)(k+m)+mn\le \ell(k+m)+2mn.
    $$
    Therefore after $t(N)$ time steps, with high probability, the learner outputs a decision tree $T_{\text{hyp}}$ satisfying 
    $$
    \dist_{\mathcal{D}}(T_{\text{hyp}},\ell\text{-}{\mathrm{\isedge}})\le \eps
    $$
    and
    \begin{align*}
        |T_{\text{hyp}}|&\le (1+\delta)\cdot \dtsize(\ell\text{-}{\mathrm{\isedge}})\tag{learner assumption}\\
        &\le (1+\delta)\cdot \left[\ell(k+m)+2mn\right]\tag{\Cref{thm:ell-isedge-intro}}
    \end{align*}
    which ensures that our algorithm correctly outputs \textsc{Yes}. 
}
\subparagraph{No case: $\VC(G)>(1+\delta')\cdot k$.}
    Assume that $\dist_{\mathcal{D}}(T_{\text{hyp}},\ell\text{-}{\mathrm{\isedge}})\le \eps<1/|\supp(\ell\text{-}D_G)|$ (otherwise the algorithm correctly outputs \textsc{No}). In particular, $\dist_{\mathcal{D}}(T_{\text{hyp}},\ell\text{-}{\mathrm{\isedge}})=0$. We would like to show that, under our assumption on $\VC(G)$, $|T_{\text{hyp}}|>(1+\delta)\cdot \left[\ell(k+m)+2mn\right]$. We start by bounding the vertex cover size of $G$:
    \begin{align*}
        \VC(G)&=\lambda \VC(G)+\frac{1-\lambda}{d} d\VC(G)\\
        &\ge \lambda \VC(G)+\frac{1-\lambda}{d} m\tag{\Cref{fact:constant degree graphs have large vc size}}\\
        &\ge \lambda \VC(G) + \left(\delta+\frac{2(1+\delta)n}{\ell}\right)m\tag{$\frac{1-\lambda}{d}>\delta+\frac{2(1+\delta)n}{\ell}$}\\
        &>(1+\delta)k+ \left(\delta+\frac{2(1+\delta)n}{\ell}\right)m.\tag{$\lambda \VC(G)>\lambda (1+\delta')k>(1+\delta)k$}
    \end{align*}
    This implies that
    \begin{equation}
    \label{eq:vc(g) lowerbound}
        \ell\VC(G)>(1+\delta)\ell k +\delta\ell m +2(1+\delta)mn.
    \end{equation}
    We can now write
    \begin{align*}
        |T_{\text{hyp}}|&\ge \dtsize(\ell\text{-}{\mathrm{\isedge}}, \ell\text{-}D_G)\tag{$T_{\text{hyp}}$ computes $\ell\text{-}{\mathrm{\isedge}}$ over $\ell\text{-}D_G$}\\
        &> \ell(\VC(G)+m)\tag{\Cref{claim:ell D_G is a certificate of large DT size}}\\
        &> (1+\delta)\ell k +(1+\delta)\ell m +2(1+\delta)mn\tag{\Cref{eq:vc(g) lowerbound}}\\
        &=(1+\delta)\cdot\left[\ell(k+m)+2mn\right]
    \end{align*}
    which ensures that our algorithm correctly outputs \textsc{No}.}
\end{proof}

%% file: fighardnessdistillation.tex
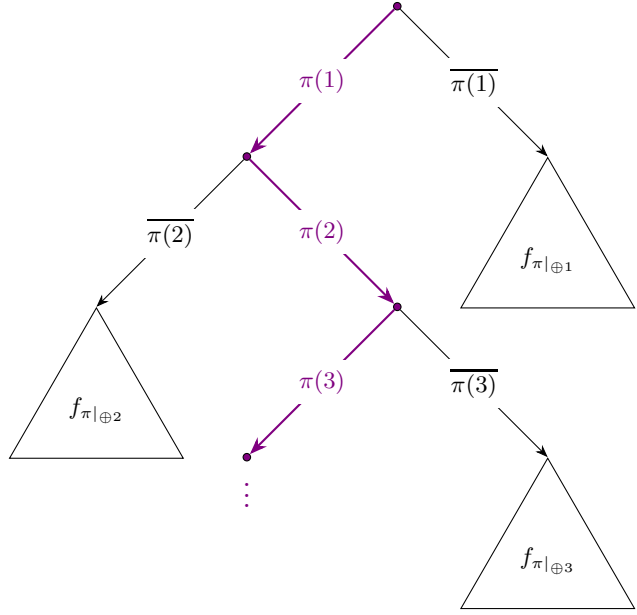
\begin{figure}[h!]
    \centering
    \begin{tikzpicture}[tips=proper]
        % points of path \pi
        \node[draw,circle,fill=violet,inner sep=1pt] (x1) at (2,2) {};
        \node[draw,circle,fill=violet,inner sep=1pt] (x2) at (0,0) {};
        \node[draw,circle,fill=violet,inner sep=1pt] (x3) at (2,-2) {};
        \node[draw,circle,fill=violet,inner sep=1pt] (x4) at (0,-4) {};

        %vdots label at the end of pi path
        \draw[color=violet] (x4) node [below] {$\vdots$};

        % arrows of path \pi
        \draw[-Stealth, violet,thick] (x1) to node[midway,fill=white!30] {\footnotesize ${\pi(1)}$} (x2);
        \draw[-Stealth, violet,thick] (x2) to node[midway,fill=white!30] {\footnotesize ${\pi(2)}$} (x3);
        \draw[-Stealth, violet,thick] (x3) to node[midway,fill=white!30] {\footnotesize ${\pi(3)}$} (x4);

        % subtrees
        \node[isosceles triangle,
            draw,
            anchor=apex,
            isosceles triangle apex angle=60,
            rotate=90,
            minimum size=2cm] (T1) at (4,0){};

        \node[isosceles triangle,
            draw,
            anchor=apex,
            isosceles triangle apex angle=60,
            rotate=90,
            minimum size=2cm] (T2) at (-2,-2){};

        \node[isosceles triangle,
            draw,
            anchor=apex,
            isosceles triangle apex angle=60,
            rotate=90,
            minimum size=2cm] (T3) at (4,-4){};
            
        % subtree labels
        \draw[color=black] (T1.center) node [] {\footnotesize $f_{\pi\vert_{\oplus 1}}$};
        \draw[color=black] (T2.center) node [] {\footnotesize $f_{\pi\vert_{\oplus 2}}$};
        \draw[color=black] (T3.center) node [] {\footnotesize $f_{\pi\vert_{\oplus 3}}$};

        % arrows to subtrees
        \draw[-Stealth] (x1) to node[midway,fill=white!30] {\footnotesize $\overline{\pi(1)}$} (T1.apex);
        \draw[-Stealth] (x2) to node[midway,fill=white!30] {\footnotesize $\overline{\pi(2)}$} (T2.apex);
        \draw[-Stealth] (x3) to node[midway,fill=white!30] {\footnotesize $\overline{\pi(3)}$} (T3.apex);
        
    \end{tikzpicture}
  \captionsetup{width=.9\linewidth}
    \caption{An illustration of hardness distillation for a function $f$. A path $\pi$ through the decision tree is highlighted in purple. This path corresponds to an ordering of a certificate for some input $x$ that follows this path. The subtrees hanging off the main path $\pi$ compute the subfunctions $f_{\pi\vert_{\oplus i}}$ where $\pi\vert_{\oplus i}$ corresponds to the path leading to the root of the subtree. The sum of the number of relevant variables of these subfunctions plus the length of the path $\pi$ lower bounds the overall size of the decision tree.}
    \label{fig:hardness distillation}
\end{figure}

%% file: figcoresetofisedge.tex
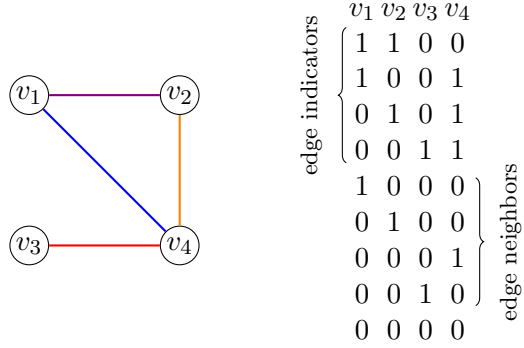
\begin{figure}[H]
\centering
\begin{tikzpicture}
\node [shape=rectangle] at (5,-1) {
        % \begin{table}[]
         \setlength{\tabcolsep}{1pt}
        \begin{tabular}{cccc}
        $v_1$ & $v_2$ & $v_3$ & $v_4$ \\
        \violet 1 & \violet 1 & \violet 0 & \violet 0 \\
        \blue 1 & \blue 0 & \blue 0 & \blue 1 \\
        \orange 0 & \orange 1 & \orange 0 & \orange 1\\
        \red 0 & \red 0 & \red 1 & \red 1\\
        1 & 0 & 0 & 0\\
        0 & 1 & 0 & 0\\
        0 & 0 & 0 & 1\\
        0 & 0 & 1 & 0\\
        0 & 0 & 0 & 0
        \end{tabular}
        % \end{table}
    };

\draw [black,decorate,decoration={brace,mirror,raise=1pt,amplitude=3pt}] (4.25,0.9) -- (4.25,-0.9) node [black,pos=0.5,rotate=90,yshift=0.5cm] {\footnotesize edge indicators};
\draw [black,decorate,decoration={brace,raise=1pt,amplitude=3pt}] (5.9,-1.1) -- (5.9,-2.8) node [black,pos=0.5,rotate=90,yshift=-0.5cm] {\footnotesize edge neighbors};

\node[shape=circle,draw=black,inner sep=1pt] (v1) at (0,0) {$v_1$};
\node[shape=circle,draw=black,inner sep=1pt] (v2) at (2,0) {$v_2$};
\node[shape=circle,draw=black,inner sep=1pt] (v3) at (0,-2) {$v_3$};
\node[shape=circle,draw=black,inner sep=1pt] (v4) at (2,-2) {$v_4$};

\path [-,thick, violet] (v1) edge node[left] {} (v2);
\path [-,thick, blue] (v1) edge node[left] {} (v4);
\path [-,thick, orange] (v2) edge node[left] {} (v4);
\path [-,thick, red] (v3) edge node[left] {} (v4);
\end{tikzpicture}
  \captionsetup{width=.9\linewidth}
\caption{Example of a graph $G$ on four vertices and the associated set of inputs $D_G\sse \zo^4$. Each row in the table corresponds to a data point in $D_G$. The first $4$ rows correspond to the edges in $G$ and are color coded to highlight which row corresponds to which edge. The next $4$ rows correspond to $1$-coordinate perturbations of the edge indicators, all of which are Hamming neighbors of edge indicators.}
\label{fig:example of G and DG}
\end{figure}

%% file: figellisedgecoreset.tex
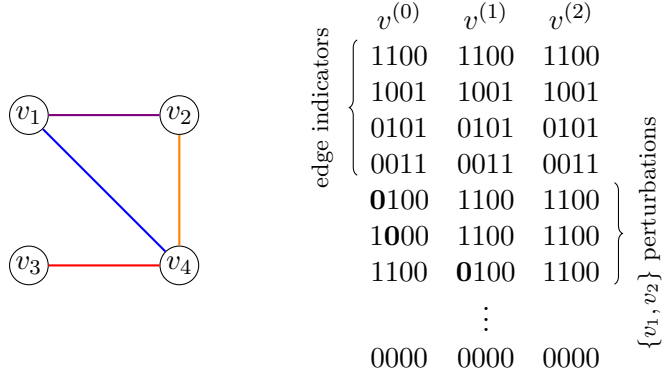
\begin{figure}[H]
\centering
\begin{tikzpicture}
\node [shape=rectangle] at (6,-1) {
        % \begin{table}[]
         \setlength{\tabcolsep}{5pt}
        \begin{tabular}{ccc}
        {$v^{(0)}$} & $v^{(1)}$ & $v^{(2)}$ \\
        \violet{1100} & \violet{1100} & \violet{1100} \\
        \blue{1001} & \blue{1001} & \blue{1001} \\
        \orange{0101} & \orange{0101} & \orange{0101} \\
        \red{0011} & \red{0011} & \red{0011} \\
        $\mathbf{0}$100 & 1100 & 1100 \\
        1$\mathbf{0}$00 & 1100 & 1100 \\
        1100 & $\mathbf{0}$100 & 1100 \\
        & $\vdots$ &\\
        0000 & 0000 & 0000
        \end{tabular}
        % \end{table}
    };
\draw [black,decorate,decoration={brace,mirror,raise=1pt,amplitude=3pt}] (4.4,1) -- (4.4,-0.8) node [black,pos=0.5,rotate=90,yshift=0.5cm] {\footnotesize edge indicators};
\draw [black,decorate,decoration={brace,raise=1pt,amplitude=3pt}] (7.75,-0.9) -- (7.75,-2.25) node [black,pos=0.5,rotate=90,yshift=-0.5cm] {\footnotesize $\{v_1,v_2\}$ perturbations};
\node[shape=circle,draw=black,inner sep=1pt] (v1) at (0,0) {$v_1$};
\node[shape=circle,draw=black,inner sep=1pt] (v2) at (2,0) {$v_2$};
\node[shape=circle,draw=black,inner sep=1pt] (v3) at (0,-2) {$v_3$};
\node[shape=circle,draw=black,inner sep=1pt] (v4) at (2,-2) {$v_4$};

\path [-,thick, violet] (v1) edge node[left] {} (v2);
\path [-,thick, blue] (v1) edge node[left] {} (v4);
\path [-,thick, orange] (v2) edge node[left] {} (v4);
\path [-,thick, red] (v3) edge node[left] {} (v4);
\end{tikzpicture}
  \captionsetup{width=.9\linewidth}
\caption{Example of a graph $G$ on four vertices and the associated set of inputs $\ell\text{-}D_G$ with $\ell=2$. The colored collection of points correspond to edge indicators. The next collection of points correspond to $1$-coordinate perturbations of the duplicated variables of the edge indicator for the edge $e=\{v_1,v_2\}$. The perturbed coordinates are bold.}
\label{fig:example of G and ell-DG}
\end{figure}

%% file: constant_error.tex
\section{Hardness for constant error} 
\subsection{Hardness of partial vertex cover}

\paragraph{Partial vertex cover.} {For a graph $G=(V,E)$ and $\alpha\in [0,1)$, an $\alpha$-partial vertex cover is subset of the vertices $C\sse V$ such that $C$ covers at least a $(1-\alpha)$-fraction of the edges. The problem $0$-partial vertex cover is the ordinary vertex cover problem. We write $\VC_\alpha(G)\in \N$ to denote the size of the smallest $\alpha$-partial vertex cover of $G$. See \Cref{fig:vc and partial vc} for an example of a partial vertex cover. The problem $\alpha$-partial $(k,k')$\textsc{-VertexCover} is to distinguish whether there exists an $\alpha$-partial vertex cover of size $\le k$ or every $\alpha$-partial vertex cover requires size $>k'$. As with ordinary vertex cover, solving this gapped problem is equivalent to approximating $\alpha$-partial vertex cover. \Cref{thm:hardness of vertex cover} implies hardness of approximating $\alpha$-partial vertex cover. It is possible to upgrade an $\alpha$-partial vertex cover to an ordinary vertex cover by augmenting it with the vertices of uncovered edges.}

\begin{fact}[Upgrading $\alpha$-partial vertex covers]
\label{fact:upgrading a partial vc}
Any $\alpha$-partial vertex cover $C$ for a graph $G$ with $m$-edges can be transformed into a vertex cover $C'$ for $G$ satisfying $|C'|\le |C|+2\alpha m$. 
\end{fact}

\input{figvcandpartialvc}

By definition, if $C$ is an $\alpha$-partial vertex cover, then $C$ leaves at most $\alpha m$ edges uncovered. Augmenting $C$ with the $\le 2\alpha m$ vertices of these uncovered edges yields a vertex cover of $G$. The size of the resulting vertex cover is $\Theta(m)$ which would be problematic if $G$ has small vertex covers. Fortunately, for constant degree graphs, $\VC(G)=\Theta(m)$ (\Cref{fact:constant degree graphs have large vc size}), so $\alpha$-partial vertex covers for these graphs are close to optimal vertex covers. This enables us to show that $\alpha$-partial vertex cover on constant degree graphs is just as hard to approximate as vertex cover. \Cref{claim:partial-vc-intro} follows by combining \Cref{claim:partial vertex cover hardness} with \Cref{thm:hardness of vertex cover}. 

\begin{claim}[Hardness of approximating $\alpha$-partial vertex cover]
    \label{claim:partial vertex cover hardness}
    For every constant $c'>1$ and $d\in\N$, there are constants $\alpha \in (0,1)$ and $c>1$ such that if there is an algorithm solving $\alpha$-partial $(k,c\cdot k)$\textsc{-VertexCover} on $n$-vertex, degree-$d$ graphs in time $t(n)$, then there is an algorithm for solving $(k,c'\cdot k)$\textsc{-VertexCover} on $n$-vertex degree-$d$ graphs in time $t(n)$. One can assume that $\alpha<\frac{1}{d+1}$. 
\end{claim}

\begin{proof}
Given $c'>1$, let $\alpha\in (0,1)$ be small enough so that $1<(1-2\alpha d)c'$ (and also small enough so that $\alpha<\frac{1}{d+1}$ for the second part of the claim) and let $c$ be any constant satisfying $1<c<(1-2\alpha d)c'$. We will solve $(k,c'\cdot k)$\textsc{-VertexCover} using an algorithm for $\alpha$-partial $(k,c\cdot k)$\textsc{-VertexCover}.

Given a graph $G$ and a parameter $k$, run the algorithm for $\alpha$-partial $(k,c\cdot k)$\textsc{-VertexCover} on $G$ and $k$. Output \textsc{Yes} if and only if the algorithm returns \textsc{Yes}. We claim that this procedure solves $(k,c'\cdot k)$\textsc{-VertexCover} on degree-$d$ graphs. For correctness, we analyze the yes and no cases separately.
\paragraph{{Yes} case: $\VC(G)\le k$.}{
    In this case, we have
    $$
    \VC_{\alpha}(G)\le \VC(G)\le k
    $$
    and so the algorithm correctly outputs \textsc{Yes}. 
}
\paragraph{{No} case: $\VC(G)>c'k$.}{
Let $m$ denote the number of edges of $G$ and let $C$ be the smallest $\alpha$-partial vertex cover of $G$. \Cref{fact:upgrading a partial vc} implies that $|C'|-2\alpha m\le |C|=\VC_\alpha(G)$ where $C'$ is a (possibly suboptimal) vertex cover for $G$. Therefore,
\begin{align*}
    \VC_\alpha(G)&\ge \VC(G)-2\alpha m\tag{\Cref{fact:upgrading a partial vc}}\\
    &\ge \VC(G)- 2\alpha d\cdot \VC(G)\tag{\Cref{fact:constant degree graphs have large vc size}}\\
    &>(1- 2\alpha d)c'k\tag{$\VC(G)>ck$ by assumption}\\
    &>ck\tag{$c<(1-2\alpha d)c'$}
\end{align*}
which means the algorithm correctly outputs \textsc{No}.}
\end{proof}

\subsection{Definition of the hard distribution}

For \Cref{thm:learning hardness inverse poly error}, we used the distribution which was uniform over the set $\ell\text{-}D_G$. This distribution has the property that the target function $\ell\text{-}\mathrm{\isedge}$ \textit{can} be approximated with subconstant error by a small decision tree. In fact, the constant function $f(x)=0$ obtains error $\le \Pr[\ell\text{-}\mathrm{\isedge}=1]\le 1/m$ in approximating $\ell\text{-}\mathrm{\isedge}$. Therefore, to obtain hardness in the constant-error regime, we need to define a new distribution, one over which the target function $\ell\text{-}\mathrm{\isedge}$ is close to balanced. To this end, we define the following distribution.

\begin{definition}[Constant-error hard distribution]
For a graph $G$ and $\ell\in\N$, the distribution $\ell\text{-}\mathcal{D}_G$ over $\zo^{n}\times (\zo^\ell)^n$ is obtained via the following experiment
\begin{itemize}
    \item with probability $1/2$ sample the all $0$s input;
    \item with probability $1/4$ sample a generalized edge indicator, $\ell\text{-}\ind[e]$ for $e\in E$ uniformly at random;
    \item with probability $1/4$ sample a $1$-coordinate perturbation of an edge indicator uniformly at random.
\end{itemize}
\end{definition}

We prove the following analogue of \Cref{claim:ell D_G is a certificate of large DT size} which shows that constant error decision trees must have large size. 
\begin{claim}
\label{claim:constant error dt size lb}
    Let $G$ be an $m$-edge graph, $\ell\in \N$, and $\alpha\in (0,1)$. If $T$ is a decision tree satisfying
    $$
    \dist_{\ell\text{-}\mathcal{D}_G}(T,\ell\text{-}\mathrm{\isedge})\le \frac{1}{16}\cdot\alpha
    $$
    then
    $$
    |T|\ge (\ell+1)\cdot\left[\VC_\alpha(G)+(1-\alpha)m\right].
    $$
\end{claim}

We first need the following lemma showing how to extract an $\alpha$-partial vertex cover from a decision tree for $\ell\text{-}\mathrm{\isedge}$.
\begin{lemma}[Obtaining an $\alpha$-partial vertex cover from a constant-error decision tree for $\ell\text{-}\mathrm{\isedge}$]
\label{lemma:alpha partial vc from constant error dt}
Let $T$ be a decision tree satisfying
$$
\dist_{\ell\text{-}\mathcal{D}_G}(T,\ell\text{-}\mathrm{\isedge})<\frac{1}{4}\alpha
$$
for any constant $\alpha\in (0,1)$. Then:
\begin{enumerate}
    \item the set $E'=\{e\in E\mid T(\ell\text{-}\ind[e])=1\}$ satisfies $|E'|\ge (1-a)m$; and
    \item if $\pi=(\overline{v}_{i_1}^{(j_1)},\ldots,\overline{v}_{i_k}^{(j_k)})$ is the path followed by $0^N$ in $T$, then for $E_\kappa=E(v_{i_\kappa}; v_{i_1},\ldots,v_{i_{\kappa-1}})$, the set of vertices
    $$
    C=\{v_{i_\kappa}\mid E'\cap E_\kappa\neq \varnothing\}
    $$
    covers all edges in $E'$. In particular, $C$ is an $\alpha$-partial vertex cover. 
\end{enumerate}
\end{lemma}

\begin{proof} We prove the two points separately.
\paragraph{First point: $|E'|\ge (1-\alpha)m$.}{The set of edges $E\setminus E'$ correspond to inputs $\ell\text{-}\ind[e]$ such that $T(\ell\text{-}\ind[e])=0$ but $\ell\text{-}\mathrm{\isedge}(\ell\text{-}\ind[e])=1$. Since each input $\ell\text{-}\ind[e]$ has mass $\frac{1}{4m}$ over $\ell\text{-}\mathcal{D}_G$, we have
\begin{align*}
    \frac{1}{4}\alpha &> \dist_{\ell\text{-}\mathcal{D}_G}(T,\ell\text{-}\mathrm{\isedge})\tag{Assumption}\\
    &\ge |E\setminus E'|\cdot \frac{1}{4m}\tag{Definition of $E'$}\\
    &=(m-|E'|)\cdot\frac{1}{4m}.
\end{align*}
}

\paragraph{Second point: $C$ is an $\alpha$-partial vertex cover.}{
Since $\dist(T,\ell\text{-}\mathrm{\isedge})<1/2$, we know that $T(0^N)=0$ and therefore the path $\pi$ terminates in a $0$-leaf. For every edge $e\in E'$, the input $\ell\text{-}\ind[e]$ must diverge from $\pi$ at some point $v_{i_\kappa}^{j_\kappa}$. This $\kappa$ then satisfies $e\in E_\kappa$ so that $e\in E'\cap E_\kappa\neq\varnothing$. It follows that $C$ covers the edge $e$. Since $E'$ constitutes at least a $(1-\alpha)$-fraction of the edges, $C$ is an $\alpha$-partial vertex cover.
}
\end{proof}

\begin{proof}[Proof of \Cref{claim:constant error dt size lb}]
Let $T$ be any decision tree such that
$$
\dist_{\ell\text{-}\mathcal{D}_G}(T,\ell\text{-}\mathrm{\isedge})\le \frac{1}{16}\alpha.
$$
In particular, $T$ satisfies the conditions of \Cref{lemma:alpha partial vc from constant error dt}. Let $E',\pi, E_\kappa$, and $C$ be as in the statement of \Cref{lemma:alpha partial vc from constant error dt}. For each $v_{i_\kappa}\in C$, we define 
$$
R(v_{i_\kappa})\coloneqq \textsc{Dup}(v_{i_\kappa})\cup\left\{v_{i_\kappa}^{(0)}\right\}\setminus\left\{v_{i_\kappa}^{(j_\kappa)}\right\}\cup\bigcup_{\{v_{i_\kappa},v\}\in E'\cap E_\kappa} \textsc{Dup}(v)\cup\left\{v^{(0)}\right\}.
$$
Furthermore, let $T_\kappa$ be the subtree which is the right child of $\pi(\kappa)$. That is $T_\kappa$ is the subtree of $T$ which catches all of the inputs $\ell\text{-}\ind[e]$ for $e\in E'\cap E_\kappa$. Recall from the proof of \Cref{lemma:covered edges lemma extended for ell-isedge} that the variables in $R(v_{i_\kappa})$ are all relevant for the subfunction $\ell\text{-}\mathrm{\isedge}_{\pi\vert_{\oplus \kappa}}$. Each such relevant variable which is \textit{not} queried in the subtree $T_\kappa$ results in an error. For example, if $T_{\kappa}$ does not query a variable $v'\in \textsc{Dup}(v_{i_\kappa})\cup\left\{v_{i_\kappa}^{(0)}\right\}\setminus\left\{v_{i_\kappa}^{(j_\kappa)}\right\}$, then the string $\ell\text{-}\ind[e]^{\oplus v'}$ where $e\in E'\cap E_\kappa$ is classified as $1$ by $T$:
$$
T(\ell\text{-}\ind[e]^{\oplus v'})=T_\kappa(\ell\text{-}\ind[e]^{\oplus v'})=T_\kappa(\ell\text{-}\ind[e])=1
$$
whereas $\ell\text{-}\mathrm{\isedge}(\ell\text{-}\ind[e]^{\oplus v'})=0$. Thus, each relevant variable which is \textit{not} queried in the subtree $T_\kappa$ results in a $0$-input being misclassified as $1$. Each such misclassification contributes $\Pr_{\ell\text{-}\mathcal{D}_G}[\ell\text{-}\ind[e]^{\oplus v'}]\ge \frac{1}{4}\cdot \frac{1}{m(2\ell+2)}$ to $\dist_{\ell\text{-}\mathcal{D}_G}(T,\ell\text{-}\mathrm{\isedge})$. Therefore, we can write
\begin{align*}
    \frac{1}{16}\alpha &\ge \dist_{\ell\text{-}\mathcal{D}_G}(T,\ell\text{-}\mathrm{\isedge})\\
    &\ge \left(\sum_{v_{i_\kappa}\in C}|R(v_{i_\kappa})|-|T_{\kappa}|\right)\cdot \frac{1}{8(m\ell+m)}+\left(m-|E'|\right)\cdot \frac{1}{4m}\\
    &\ge \left(\sum_{v_{i_\kappa}\in C}|R(v_{i_\kappa})|-|T_{\kappa}|\right)\cdot \frac{1}{16\ell m}+\left(m-|E'|\right)\cdot \frac{1}{4m}\tag{$m\le \ell m$}
\end{align*}
where the quantity $\sum_{v_{i_\kappa}\in C}|R(v_{i_\kappa})|-|T_{\kappa}|$ counts how many $0$-inputs are misclassified as $1$ by $T$ and $m-|E'|$ counts how many $1$-inputs are misclassified as $0$. These quantities are weighted by the respective masses of each type of input over $\ell\text{-}\mathcal{D}_G$. Rearranging gives the lower bound:
\begin{align*}
    \sum_{v_{i_\kappa}\in C}|T_\kappa|&\ge 4\ell(m-|E'|)-\alpha \ell m + \sum_{v_{i_\kappa}\in C}|R(v_{i_\kappa})|\\
    &\ge 4\ell(m-|E'|)-\alpha \ell m +\ell|C|+\sum_{v_{i_\kappa}\in C}(\ell+1) |E'\cap E_\kappa|\tag{$|\textsc{Dup}(v)=\ell|$}\\
    &=4\ell(m-|E'|)-\alpha \ell m +\ell|C|+(\ell+1) |E'|\tag{$\{E'\cap E_\kappa\}$ partitions $E'$}\\
    &\ge 4\ell(m-|E'|)-\alpha \ell m +\VC_\alpha(G)\ell+(\ell+1) |E'|\tag{$C$ is an $\alpha$-partial vertex cover}\\
    &\ge \ell\VC_\alpha(G)+(1-\alpha)(\ell+1)m \tag{$|E'|\le m$}.
\end{align*}
Therefore, since the $T_{\kappa}$ and $\pi$ are all disjoint parts of $T$: 
\begin{align*}
    |T|&\ge |\pi|+\sum_{v_{i_\kappa}\in C}|T_\kappa|\\
    &\ge |C|+\sum_{v_{i_\kappa}\in C}|T_\kappa|\tag{Definition of $C$}\\
    &\ge \VC_{\alpha}(G)+\sum_{v_{i_\kappa}\in C}|T_\kappa|\tag{C is an $\alpha$-partial vertex cover}\\
    &\ge (\ell+1)\left[\VC_{\alpha}(G)+(1-\alpha)m\right]
\end{align*}
which completes the proof.
\end{proof}

\subsection{Learning consequence for constant-error: Proof of \Cref{thm:main-intro}}

\begin{theorem}[Hardness of learning DTs with constant error]
\label{thm:learning hardness constant error}
For all constants $\delta'>0$, $d\in\N$, and $\alpha<\frac{1}{d+1}$, there is a sufficiently small constant $\delta>0$ such that the following holds. If \textsc{DT-Learn}$(n,s,(1+\delta)\cdot s,\eps)$ with $s=O(n)$ and $\eps=\Theta(1)$ can be solved in randomized time $t(n)$, then $\alpha$-\textsc{PartialVertexCover}$(k,(1+\delta')k)$ on degree-$d$ graphs can be solved in time $O(n^2t(n^2))$. 
\end{theorem}

The proof of this theorem is similar to that of \Cref{thm:learning hardness inverse poly error}. The main difference is that our lower bound on the decision tree size of $\ell\text{-}\mathrm{\isedge}$ in the constant-error regime is quantitatively weaker than that of \Cref{claim:ell D_G is a certificate of large DT size}. We will need to make the appropriate adjustments to the approximation factor of the \textsc{DT-Learner} in order to tolerate the weaker lower bound.

\begin{proof}
    Let $\delta'>0$, $d\in\N$, and $\alpha<\frac{1}{d+1}$ be given. The assumption that $\alpha<\frac{1}{d+1}$ implies $\alpha<\frac{1-\alpha}{d}$. Therefore, we can fix some $\lambda<1$ large enough so that $\lambda (1+\delta')>1$ \textit{and} $\alpha<\frac{(1-\lambda)(1-\alpha)}{d}$. Let $\delta>0$ be any constant satisfying $\frac{(1-\lambda)(1-\alpha)}{d}>\delta+\alpha$ and $(1+\delta)<\lambda(1+\delta')$. We will solve $\alpha$-\textsc{PartialVertexCover}$(k,(1+\delta')k)$ using an algorithm for \textsc{DT-Learn}$(n,s,(1+\delta)\cdot s,\eps)$.
\paragraph{The reduction.}{Fix $\ell=\Theta(n)$ large enough so that 
\begin{equation}
\label{eq:l guarantee}
    \frac{(1-\lambda)(1-\alpha)}{d}>\delta+\alpha+\frac{2(1+\delta)n}{\ell}.
\end{equation}
Such an $\ell$ exists by our assumption that $\frac{1-\lambda}{d}>\delta+\alpha$. As in \Cref{thm:learning hardness inverse poly error} our target function will be $\ell\text{-}{\mathrm{\isedge}}:\zo^N\to\zo$ for $N=n+\ell n=\Theta(n^2)$. Our distribution will be $\ell\text{-}\mathcal{D}_G$ and we fix $\eps<\frac{1}{16}\alpha=\Theta(1)$ and $s=\ell(k+m)+2mn=O(N)$. Run the same procedure as in \Cref{fig:solving vc with dt learn} where the distribution $\mathcal{D}$ is $\ell\text{-}\mathcal{D}_G$. 
}
\paragraph{Runtime.}{
As in the proof of \Cref{thm:learning hardness inverse poly error}, queries and random samples for $\ell\text{-}{\mathrm{\isedge}}$ can be handled in $O(N)$ time. Thus running the learner requires $O(N\cdot t(N))$ time. Computing the error $\dist_{\ell\text{-}\mathcal{D}_G}(T_{\text{hyp}},\ell\text{-}{\mathrm{\isedge}})$ takes $O(N^2)$ time. The overall runtime is therefore $O(N\cdot t(N))$ which is $O(n^2\cdot t(n^2))$.
}

\paragraph{Correctness.}{We analyze the \textsc{Yes} case and \textsc{No} case separately.
\subparagraph{{Yes} case: $\VC_\alpha(G)\le k$.}{
This case is identical to the \textsc{Yes} case in \Cref{thm:learning hardness inverse poly error}. So our algorithm correctly outputs $\textsc{Yes}.$
}
\subparagraph{{No} case: $\VC_\alpha(G)>(1+\delta')k$.}{
    Assume that $\dist_{\ell\text{-}\mathcal{D}_G}(T_{\text{hyp}},\ell\text{-}{\mathrm{\isedge}})\le \eps<\frac{1}{16}\alpha$ (otherwise our algorithm correctly outptus \textsc{No}). We would like to show that $|T_{\text{hyp}}|>(1+\delta)\cdot\left[\ell(k+m)+2mn\right]$. We start by bounding $\alpha$-partial vertex cover size of $G$:
    \begin{align*}
        \VC_\alpha(G)&=\lambda \VC_\alpha(G)+\frac{1-\lambda}{d} d\VC_\alpha(G)\\
        &\ge \lambda  \VC_\alpha(G)+\frac{(1-\lambda)(1-\alpha)}{d}m\tag{$d\VC_\alpha(G)\ge (1-\alpha)m$ for degree $d$ graphs}\\
        &\ge \lambda  \VC_\alpha(G)+\left(\delta+\alpha+\frac{2(1+\delta)n}{\ell}\right)m\tag{\Cref{eq:l guarantee}}\\
        &\ge (1+\delta)k+\left(\delta+\alpha+\frac{2(1+\delta)n}{\ell}\right)m\tag{$\lambda\VC_\alpha(G)>\lambda (1+\delta')k>(1+\delta)k$}.
    \end{align*}
    Rearranging gives
    \begin{equation}
    \label{eq:ell vc alpha lb}
        \ell\VC_\alpha(G)\ge (1+\delta)k\ell+(\delta+\alpha)m\ell+2(1+\delta)mn.
    \end{equation}
    Therefore,
    \begin{align*}
        |T_{\text{hyp}}|&> \ell(\VC_\alpha(G)+(1-\alpha)m)\tag{\Cref{claim:constant error dt size lb}}\\
        &> (1+\delta)k\ell+(\delta+\alpha)m\ell+2(1+\delta)mn+ (1-\alpha)m\ell\tag{\Cref{eq:ell vc alpha lb}}\\
        &=(1+\delta)\cdot\left[\ell(k+m)+2mn\right]
    \end{align*}
which ensures that our algorithm correctly outputs $\textsc{No}.$
}
}
\end{proof}

\begin{remark}[Implications for testing decision trees]
The above proof of \Cref{thm:learning hardness constant error} and the proof of \Cref{thm:learning hardness inverse poly error} actually prove hardness of \textit{testing} decision tree size. Specifically, the proof of \Cref{thm:learning hardness constant error} shows that any tester which can distinguish whether a target function $f$ is a size-$s$ decision tree or is $\Omega(1)$-far from every size-$s$ decision tree over a distribution $\mathcal{D}$ can also approximate {\sc PartialVertexCover}. Therefore, the problem of distribution-free testing decision tree size is also NP-hard.
\end{remark}

\begin{proof}[Proof of \Cref{thm:main-intro}]
If there were an algorithm for learning decision trees which satisfies the constraints of \Cref{thm:main-intro}, then \Cref{thm:learning hardness constant error} shows that $\alpha$-\textsc{PartialVertexCover} can be solved in \text{RTIME}$(n^2t(n^2))$. \Cref{thm:hardness of vertex cover} and~\Cref{claim:partial vertex cover hardness} then imply that \SAT\ can be solved in randomized time $O(n^2\polylog n\cdot t(n^2\polylog n))$.
\end{proof}

%% file: figvcandpartialvc.tex
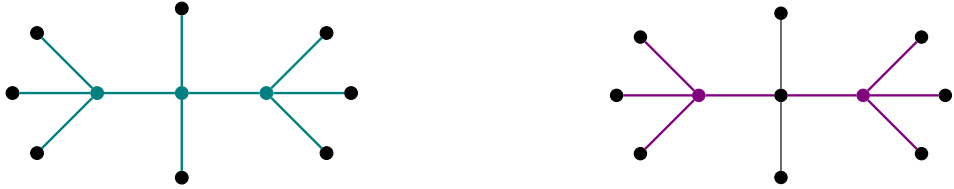
\begin{figure}[H]
\centering
\hspace*{\fill}
\begin{subfigure}{.45\textwidth}
  \centering
  \resizebox{0.65\textwidth}{!}{%
      \begin{tikzpicture}[
    ]
        \centering
        \node[shape=circle,draw=black,inner sep=1.5pt,fill=black] (WW) at (-1,0) {};
        \node[shape=circle,draw=black,inner sep=1.5pt,fill=black] (WN) at (-0.71,0.71) {};
        \node[shape=circle,draw=black,inner sep=1.5pt,fill=black] (WS) at (-0.71,-0.71) {};
        \node[shape=circle,draw=teal,inner sep=1.5pt,fill=teal] (W) at (0,0) {};
        
        \node[shape=circle,draw=teal,inner sep=1.5pt,fill=teal] (C) at (1,0) {};
        \node[shape=circle,draw=black,inner sep=1.5pt,fill=black] (CN) at (1,1) {};
        \node[shape=circle,draw=black,inner sep=1.5pt,fill=black] (CS) at (1,-1) {};
        
        \node[shape=circle,draw=teal,inner sep=1.5pt,fill=teal] (E) at (2,0) {};
        \node[shape=circle,draw=black,inner sep=1.5pt,fill=black] (EE) at (3,0) {};
        \node[shape=circle,draw=black,inner sep=1.5pt,fill=black] (EN) at (2.71,0.71) {};
        \node[shape=circle,draw=black,inner sep=1.5pt,fill=black] (ES) at (2.71,-0.71) {};
        
        \path [-,teal,thick] (WW) edge node[left] {} (W);
        \path [-,teal,thick] (WN) edge node[left] {} (W);
        \path [-,teal,thick] (WS) edge node[left] {} (W);
        \path [-,teal,thick] (C) edge node[left] {} (W);
        
        \path [-,teal,thick] (CN) edge node[left] {} (C);
        \path [-,teal,thick] (CS) edge node[left] {} (C);
        \path [-,teal,thick] (E) edge node[left] {} (C);
        
        \path [-,teal,thick] (EE) edge node[left] {} (E);
        \path [-,teal,thick] (EN) edge node[left] {} (E);
        \path [-,teal,thick] (ES) edge node[left] {} (E);
    \end{tikzpicture}
    }
  \caption{A vertex cover and its covered edges highlighted in teal}
  \label{fig:vc-cover}
\end{subfigure}%
\hfill
\begin{subfigure}{.45\textwidth}
  \centering
  \resizebox{0.65\textwidth}{!}{
          \begin{tikzpicture}[
        ]
            \centering
            \node[shape=circle,draw=black,inner sep=1.5pt,fill=black] (WW) at (-1,0) {};
            \node[shape=circle,draw=black,inner sep=1.5pt,fill=black] (WN) at (-0.71,0.71) {};
            \node[shape=circle,draw=black,inner sep=1.5pt,fill=black] (WS) at (-0.71,-0.71) {};
            \node[shape=circle,draw=violet,inner sep=1.5pt,fill=violet] (W) at (0,0) {};
            
            \node[shape=circle,draw=black,inner sep=1.5pt,fill=black] (C) at (1,0) {};
            \node[shape=circle,draw=black,inner sep=1.5pt,fill=black] (CN) at (1,1) {};
            \node[shape=circle,draw=black,inner sep=1.5pt,fill=black] (CS) at (1,-1) {};

            \node[shape=circle,draw=violet,inner sep=1.5pt,fill=violet] (E) at (2,0) {};
            \node[shape=circle,draw=black,inner sep=1.5pt,fill=black] (EE) at (3,0) {};
            \node[shape=circle,draw=black,inner sep=1.5pt,fill=black] (EN) at (2.71,0.71) {};
            \node[shape=circle,draw=black,inner sep=1.5pt,fill=black] (ES) at (2.71,-0.71) {};
            
            \path [-,violet,thick] (WW) edge node[left] {} (W);
            \path [-,violet,thick] (WN) edge node[left] {} (W);
            \path [-,violet,thick] (WS) edge node[left] {} (W);
            \path [-,violet,thick] (C) edge node[left] {} (W);
            
            \path [-] (CN) edge node[left] {} (C);
            \path [-] (CS) edge node[left] {} (C);
            \path [-,violet,thick] (E) edge node[left] {} (C);
            
            \path [-,violet,thick] (EE) edge node[left] {} (E);
            \path [-,violet,thick] (EN) edge node[left] {} (E);
            \path [-,violet,thick] (ES) edge node[left] {} (E);
        \end{tikzpicture}
    }
  \caption{A $\frac{1}{5}$-partial vertex cover and its covered edges highlighted in purple}
  \label{fig:partial vc}
\end{subfigure}
\hspace*{\fill}
\caption{A graph $G=(V,E))$ with $10$ edges having $\VC(G)=3$ and $\VC_{1/5}(G)=2$.}
\label{fig:vc and partial vc}
\end{figure}